\documentclass[a4paper]{llncs}

\usepackage{enumerate}
\usepackage{fullpage}
\usepackage[english]{babel}
\usepackage[utf8]{inputenc}
\usepackage{cite}
\usepackage[fleqn]{amsmath}
\usepackage{enumerate}
\usepackage{enumitem}
\usepackage{amsfonts}
\usepackage{xspace}
\usepackage[colorlinks=false]{hyperref}
\usepackage{bookmark}
\usepackage{xcolor}
\usepackage{soul}
\newtheorem{observation}{Observation}

\usepackage{mathtools}
\usepackage{multicol}
\usepackage{tikz}
\spnewtheorem{clm}{Claim}{\itshape}{\rmfamily}

\newcommand{\cross}{\mathbin{\tikz [x=1.4ex,y=1.4ex,line width=.2ex] \draw (0,0) -- (1,1) (0,1) -- (1,0);}}%
\DeclarePairedDelimiter{\ceil}{\lceil}{\rceil}
\DeclarePairedDelimiter{\floor}{\lfloor}{\rfloor}

\title{An approximation algorithm for Uniform Capacitated $k$-Median problem with $1+\epsilon$ capacity violation}

\author{Jarosław Byrka\inst{1} \and
Bartosz Rybicki\thanks{Research supported by NCN 2012/07/N/ST6/03068 grant} \inst{1} \and Sumedha Uniyal\thanks{Partially supported by the ERC StG project NEWNET no.~279352.} \inst{2}} 
\authorrunning{Jaros?aw Byrka et al.} 

\tocauthor{JarosÅaw Byrka, Bartosz Rybicki, Sumedha Uniyal}
\institute{Institute of Computer Science, University of Wrocław, Poland,\\
\email{\{jby,bry\}@cs.uni.wroc.pl} 
\and
IDSIA, University of Lugano, Switzerland,\\
\email{sumedha@idsia.ch} 
}

\date{\today}

\begin{document}
\maketitle

\begin{abstract}
We study the Capacitated $k$-Median problem, for which all the known constant factor approximation algorithms violate either the number of facilities or the capacities.
While the standard LP-relaxation can only be used for algorithms violating one of the two by a factor of at least two, Shi Li [SODA'15, SODA'16] gave algorithms violating the number of facilities by a factor of $1+\epsilon$ exploring properties of extended relaxations. 

In this paper we develop a constant factor approximation algorithm for Uniform Capacitated $k$-Median violating only the capacities by a factor of $1+\epsilon$. The algorithm is based on a configuration LP. Unlike in the algorithms violating the number of facilities, we cannot simply open extra few facilities at selected locations. Instead, our algorithm decides about the facility openings in a carefully designed dependent rounding process.
\end{abstract}

\section{Introduction}

In capacitated $k$-median we are given a set of potential facilities $F$, capacity $u_i\in \mathbb{N}^+$ for each facility $i \in F$, a set of clients $C$, a metric distance function $d$ on $C \cup F$ and an integer $k$. The goal is to find a subset $F' \subseteq F$ of $k$ facilities to open and an assignment $\sigma: C \rightarrow F'$ of clients to the open facilities such that $|\sigma^{-1}(i)| \leq u_i$ for every $i \in F'$, so as to minimize the connection cost $\sum_{j \in C}d(j, \sigma(j))$. In the uniform capacity case, $u_i=u, \forall i \in F$.

The standard $k$-median problem, where there is no restriction on the number of clients served by a facility, can be approximated up to a constant factor~\cite{charikar1999constant, arya2004local}. The current best is the $(2.675+\epsilon)$-approximation algorithm of Byrka et al.~\cite{byrka2015improved}, which is a result of optimizing a part of the algorithm of Li and Svensson~\cite{li2013approximating}.

Capacitated $k$-median is among the few remaining fundamental optimization problems for which it is not clear if there exist constant factor approximation algorithms. All the known algorithms violate either the number of facilities or the capacities. In particular, already the algorithm of Charikar et al.~\cite{charikar1999constant} gave 16-approximate solution for uniform capacitated $k$-median violating the capacities by a factor of 3.
Then Chuzhoy and Rabani~\cite{chuzhoy2005approximating} considered general capacities and gave a 50-approximation algorithm violating capacities by a factor of 40.

Perhaps the difficulty is related to the unbounded integrality gap of the standard LP relaxation. To obtain integral solutions that are bounded w.r.t. a fractional solution to the standard LP, one has to either allow the integral solution to open twice more facilities or to violate the capacities by a factor of two. Recently, LP-rounding algorithms essentially matching these limits were obtained~\cite{aardal2015approximation, BFRS15}.

Next, Li broke this integrality gap barrier by giving a constant factor algorithm for uniform capacitated $k$-median by opening $(1+\epsilon)k$ facilities~\cite{Li15uniform}. The algorithm is based on rounding a fractional solution to an extended LP. Most recently, he gave an algorithm working with general capacities and still opening $(1+\epsilon)k$ facilities~\cite{Li16non_uniform}. This new algorithm is based on an even stronger configuration LP. Notably, each of the extended linear programs is not solved exactly by the algorithms, but rather a clever "round-or-separate" technique is applied. This technique was previously used in the context of capacitated facility location in~\cite{an2014lp}.
It essentially allows not to solve the strong LP upfront, but only to detect violated constraints during the rounding process. If a violated constraint is detected, it is returned back to the "feasibility-checking" ellipsoid algorithm. While it is not clear if the strong LP with all the constraints can be solved efficiently, it can be shown that the above described process terminates in polynomial time, see~\cite{an2014lp}.

\subsection{Our results and techniques}

We give an algorithm for uniform capacitated $k$-median rounding a fractional solution to the configuration LP, from~\cite{Li16non_uniform}, via the "round-or-separate" technique. We obtain a constant factor approximate integral solution violating capacities by a factor of $1+\epsilon$. We utilize the power of the configuration LP in effectively rounding small size facility sets, and combine it with a careful dependent rounding to coordinate the opening between these small sets. The main result of this paper is described in the following theorem.

\begin{theorem}
\label{thm:general}
There is a bi-factor randomized rounding algorithm for hard uniform capacitated $k$-median problem, with $O(1/\epsilon^2)$-approximation under $1+\epsilon$ capacity violation.
\end{theorem}

Our algorithm utilizes the white-grey-black tree structure from~\cite{Li16non_uniform}, but the following rounding steps are quite different. In particular, the handling of the small "black components" differs. While aiming for a solution opening $(1+\epsilon)k$ facilities Li~\cite{Li16non_uniform} can treat each black component independently, we are forced to precisely determine the number of open facilities.
Hence we cannot allow a single black component to individually decide to open more facilities than in the fractional solution. Instead, we first use a preprocessing step which we call \emph{massage} that reduces the variance in the number of open facilities in the fractional solution within each "black component". Then we use a form of a pipeage rounding between the "black components", that precisely preserves the total number of open facilities. The tree structure is used to route the demand to the eventually open facilities.

\section{Linear Program}

The following is the basic LP relaxation for the problem:

\begin{multicols}{2}
\noindent
\begin{align}
\min \, & \textstyle \sum_{i\in F,j\in C}d(i,j) x_{i,j} &\hspace{-2em} \mbox{s.t.}\hspace{-2em} \nonumber \\
        & \textstyle \sum_{i\in F}y_i = k;\hspace{-4em}~\label{blp:k} \\ 
        & \textstyle \sum_{i\in F}x_{i,j} = 1 &\hspace{-4em} \forall j\in C;~\label{blp:connect}
\end{align}
\columnbreak
\begin{align}
\hfill \hspace{-1em}&&\hspace{-1em} & \hspace{-2em} \mbox{(Basic LP)} & \nonumber \\
\hspace{-1em}& \textstyle \sum_{j\in C}x_{i,j} \leq u_i y_i \hspace{-5em}&\hspace{-2em}\forall i\in F;~\label{blp:capacity} \\
\hspace{-1em}& 0 \leq x_{i,j} \leq y_{j} \leq 1 &\hspace{-1em} \forall i\in F, j\in C.\hspace{-1em}~\label{blp:range}
\end{align}
\noindent
\end{multicols}

In the above LP, $y_i$ indicates whether facility $i$ is open or not, and $x_{i,j}$ indicates whether client $j$ is connected to facility $i$ or not. Constraint (\ref{blp:k}) is the cardinality constraint that exactly $k$ facilities are open, Constraint (\ref{blp:connect}) requires every client $j$ to be connected, Constraint (\ref{blp:range}) says that a client can only be connected to an open facility and Constraint (\ref{blp:capacity}) is the capacity constraint. 

The basic LP has an unbounded integrality gap even if we allow to violate the cardinality or the capacity constraints by $2-\epsilon$. To overcome this gap, Li introduced in \cite{Li16non_uniform} a stronger LP called the Configuration LP and got a constant approximation algorithm by opening $(1+\epsilon)k$ facilities. 

To formulate the configuration LP constraints, let us fix a set $B \subseteq F$ of facilities. Let $\mathcal{S}=\{S \subseteq B: |S| \leq \ell_1\}$ and $\tilde{\mathcal{S}}=\mathcal{S} \cup \{\bot\}$, where $\ell_1$ is some constant we will define later and $\bot$ stands for "any subset of B with size greather than $\ell_1$''. We treat set $\bot$ as a set that contains all the facilities $i \in B$. For every $S\in \tilde{\mathcal{S}}$, let $z_S^B$ be an indicator variable corresponding to the event that the set of open facilities in $B$ is exactly $S$ and $z_{\bot}^B$ captures the event that the number of facilities open in $B$ is more than $\ell_1$. For every $S\in \tilde{\mathcal{S}}$ and $i \in S$, $z_{S,i}^B$ indicates the event that set $S$ is open and $i$ is open as well. Notice that when $i\in S \neq \bot$, we always have $z_{S,i}=z_S^B$. For every $S\in \tilde{\mathcal{S}}$, $i \in S$ and $j \in C$, $z_{S,i,j}^B$ indicates the event that $z_{S,i}^B=1$ and $j$ is connected to $i$. The following are valid constraints for any feasible integral solution.

\vspace{-0.7cm}
\begin{multicols}{2}
\begin{align}
\noindent
\sum_{S\in \tilde{\mathcal{S}}}  {z}_S^B = 1;&\label{clp:6} \\ 
\sum_{S\in \tilde{\mathcal{S}}: i\in S}z_{S, i}^B = y_i & \forall i\in B;~\label{clp:7} \\
\sum_{S \in \tilde{\mathcal{S}}: i\in S}z_{S, i, j}^B = x_{i,j} & \forall i\in B, j \in C;\hspace{-2em}~\label{clp:8} \\
0 \leq z_{S,i,j}^B \leq z_{S,i}^B \leq z_S^B \, & \, \forall S \in \tilde{\mathcal{S}}, i \in S, j\in C;~\label{clp:9}
\end{align}
\columnbreak
\begin{align}
&&\nonumber\\
z_{S,i}^B = z_S^B, & \forall S \in \mathcal{S}, i \in S;~\label{clp:10} \\
\sum_{i\in S}z_{S, i, j}^B &\leq z_S^B  \forall S \in \tilde{\mathcal{S}}, j \in C;~\label{clp:11} \\
\sum_{j \in C}z_{S, i, j}^B & \leq u_i z_{S,i}^B  \forall S \in \tilde{\mathcal{S}}, i \in S;\hspace{-1em}~\label{clp:12} \\
\sum_{i \in B} z_{\bot,i}^B & \geq {\ell}_1 z_{\bot}^B.~\label{clp:13} 
\end{align}
\end{multicols}
\vspace{-0.7cm}

Constraint (\ref{clp:6}) says that exactly one set $S\in \tilde{\mathcal{S}}$ is open. Constraint (\ref{clp:7}) says that if facility $i$ is open then $z_{S,i}^B=1$ for exactly one set $S\in \tilde{\mathcal{S}}$. Constraint (\ref{clp:8}) says that if $j$ is connected to $i$ then $z_{S,i,j}^B=1$ for exactly one set $S\in \tilde{\mathcal{S}}$. Constraint (\ref{clp:11}) says that if $z_S^B=1$, then $j$ can be connected to at most $1$ facility in $S$. Constraint (\ref{clp:12}) is the capacity constraint. Constraint (\ref{clp:13}) says that if $z_{\bot}^B=1$, then at least $\ell_1$ facilities in $B$ are open.

The configuration LP is obtained by adding the above set of constraints for all subsets $B \subseteq F$. As there are exponentially many sets $B$, we do not know how to solve this LP. But given a fractional solution $(x,y)$, for a fixed set $B$, we can construct the values of the set of variables $z$ (see \cite{Li16non_uniform}) and also check the constraints in polynomial time since the total number of variables and constraints is $n^{O(\ell_1)}$. We apply method that has been used in, e.g., \cite{Li15uniform, Li16non_uniform}. Given a fractional solution $(x,y)$ to the basic LP relaxation, our rounding algorithm either constructs an integral solution with the desired properties or outputs a set $B\subseteq F$ for which one of the Constraints (\ref{clp:6}) to (\ref{clp:13}) is infeasible. In the latter case, we can find a violating constraint and feedback it to the ellipsoid method.

\section{Rounding Algorithm}

Focus on an optimal fractional solution $(x,y)$ to the basic LP. Let $d_{av}(j) = \sum_{i \in F}d(i,j) x_{i,j}$ be the connection cost client $j \in C$. Note that the value of the LP solution $(x,y)$ is $\mathbf{LP} := \sum_{(i,j) \in F \cross C}d(i,j)x_{i,j} = \sum_{j\in C}d_{av}(j)$. For any set $F' \subseteq F$ of facilities, let $y_{F'} := y(F') := \sum_{i \in F'}y_i$ be the volume of the set $F'$. For any set $F'\subseteq F$ and $C' \subseteq C$ of clients, let $x_{F', C'}:=\sum_{(i,j) \in F'\cross C'}x_{i,j}$. Also let, $x_{i,C'}:=x_{\{i\}, C'}$ and $x_{F',j}:=x_{F',\{j\}}$.
\begin{definition}
Let $D_i:=\sum_{j \in C} x_{i,j}d(i, j)$ and $D'_i:=\sum_{j \in C} x_{i,j}d_{av}(j)$ for each $i \in F$. Let $D_S:=D(S):=\sum_{i \in S} D_i$ and $D'_S:=D'(S):=\sum_{i \in S} D'_i$ for every $S \subseteq F$, Obviously $D_F=D'_F=\mathbf{LP}$.
\end{definition}

First we will partition facilities into clusters (as done in \cite{BFRS15},\cite{Li15uniform}). Each cluster will have a client $v$ as its \emph{representative}. We denote the set of cluster representatives by $R$. Each cluster will contain the set of facilities nearest to a representative $v \in R$ and the fractional number of open facilities in each cluster will be bounded below by $1 - \frac{1}{\ell}$. Let $U_v$ be the set of facilities in the cluster corresponding to representative $v \in R$. For any set $J \subseteq R$ of representatives, we use $U_J := U(J) = \bigcup_{v\in J} U_v$. Constants $\ell:=O(1/\epsilon)$ and $\ell_1:=\ell^2$ are integers, which we will define later. Since the clustering procedure is the same as in \cite{Li15uniform, BFRS15}, we omit. The following Claim (see Claim 4.1 in \cite{Li15uniform}) captures the key properties of the clustering procedure.

\begin{clm}
\label{claim:clustering}
The following statements hold:

\begin{enumerate}

\item for all $v, v'\in R$, $v\neq v'$, we have $d(v,v') > 2\ell\max\{d_{av}(v), d_{av}(v')\}$;
\item for all $j \in C$,  $\exists v \in R$, such that $d_{av}(v)\leq d_{av}(j)$ and $d(v,j)\leq 2\ell d_{av}(j)$;
\item $y_{U_v} \geq 1-1/\ell$ for every $v \in R$;
\item for any $v \in R$, $i\in U_v$ and $j \in C$, we have $d(i,v) \leq d(i,j) + 2\ell d_{av}(j)$.
\end{enumerate}
\end{clm}

We partition the set of representatives $R$ build a tree and color its edges in the same way as Li \cite{Li16non_uniform}. To partition $R$, we run the Kruskal's algorithm to find a minimum spanning tree of $R$. In the Kruskal's algorithm we maintain a partition $\mathcal{J}$ of $R$ and the set of selected edges $E_{MST}$. Initially $\mathcal{J} = \{\{v\} : v \in R\}$ and $E_{MST}$ is empty. The length of each edge $(u, v) \in {R \choose 2}$ is the distance between $u$ and $v$. We sort all edges in ${R \choose 2}$ by length, breaking ties in an arbitrary way. For each edge $(u ,v)$ in this order if $u$ and $v$ are not in the same group in $\mathcal{J}$, we merge the two groups and add edge $(u, v)$ to $E_{MST}$.

We now color edges of $E_{MST}$. For every $v \in R$, we know that $y(U_v) \geq 1 - 1/\ell$. For any subset of representatives $J \subseteq R$ we say that $S$ is big if $y(U_J) \geq \ell$ and small otherwise. For each edge $e \in E_{MST}$ we consider the step in which edge $e=(u,v)$ was added by the Kruskal's algorithm to MST. After the iteration we merge groups $J_u$ (containing $u$) and $J_v$ (containing $v$) to one group $J_u \cup J_v$. If both $J_u$ and $J_v$ are small, then we paint edge $e$ in black. If both are big, we paint the edge $e$ white. Otherwise if one is small and the other is big then we direct the edge $e$ towards the big group and paint it grey.

Consider only the black edges from $E_{MST}$. We define a black component of MST as a connected component in this graph. The following claim (see Claim 4.1 in \cite{Li16non_uniform}) is a consequence of the fact that $J \subseteq R$ appears as a group at some step of the Kruskal's algorithm.

\begin{clm}
\label{claim:black_edges}
	Let $J$ be a black component, then for every black edge $(u, v)$ in ${J \choose 2}$, we have $d(u, v) \leq d(J, R \setminus J)$.
\end{clm}

We contract all the black components and remove all the white edges from MST. The obtained graph $\Upsilon$ is a forest. Each node $p$ (vertex in the contracted graph) in $\Upsilon$ corresponds to a black component and each grey edge is directed. Let $J_p \subseteq R$ be the set of representatives corresponding to node $p$. Abusing the notation slightly, we define $U_p := U(J_p) = \bigcup_{v \in J_p} U_v$. Lets define $y_p:=y(U_p)$. The following lemma follows from the way in which we create our forest. Proof can be found in \cite{Li16non_uniform}.

\begin{lemma}
\label{lem:tau_properties}
 For any tree $\tau \in \Upsilon$, the following statements are true:
    \begin{enumerate}
    	\item $\tau$ has a root node $r_{\tau}$ such that all the grey edges in $\tau$ are directed towards $r_{\tau}$;
        \item $J_{r_\tau}$ is big and $J_p$ is small for all other nodes in $\tau$;
        \item in any leaf-to-root path of $\tau$, the lengths of grey edges form a non-increasing sequence;
        \item for any non-root node $p \in \tau$, the length of the grey edge in $\tau$ connecting $p$ to its parent is exactly $d(J_p, R \setminus J_p)$;
    \end{enumerate}
\end{lemma}

Consider a tree $\tau \in \Upsilon$. We group the black components of $\tau$ top down into sub-trees choosing grey edges in increasing order of their lengths until the volume of the group just exceeds $\ell$. 

\begin{definition}
A black component is called a singleton component if it contains only a single node corresponding to some $v \in R$. A singleton component which is the very root of some tree $\tau \in \Upsilon$, is called a singleton root component.
\end{definition}

\begin{observation}
\label{obs:group_vol}
Consider tree $\tau \in \Upsilon$. The root-group $\mathcal{G}$ has volume at least $\ell$. If the root-group is not a singleton root component, then it has volume at most $2\ell$. The leaf-groups might have volume smaller that $\ell$. All the other internal-groups have volume in the range $[\ell, 2\ell)$.
\end{observation}

From now on we will slightly abuse the notation and instead of $z^{U_p}$ and $x_{U_p}$ we will write $z^p$ and $x_p$, respectively. Also, we will assume that any black component $U_p$ corresponding to a node $p \in \Upsilon$ satisfies the Configuration LP Constraints (\ref{clp:6}) to (\ref{clp:13}). If not, then we find the violating constraint and we recompute the LP by applying the ellipsoid method.

In the next lemma we consider edges related to a group $\mathcal{G}$.
\begin{lemma}
\label{lem:group_edges}
For any tree $\tau$, group $\mathcal{G}$ and black component $p \in \mathcal{G}$, the following properties hold:

\begin{enumerate}

\item \label{group_edges1} the total number of grey or black edges within $\mathcal{G}$ is at most $O(\ell)$;
\item \label{group_edges2} any grey edge entering $\mathcal{G}$ is longer (or equal) to any grey or black edge in $\mathcal{G}$;
\item \label{dist_to_root1} the total length of the path (including both grey and black edges), from any node $v \in J_p$ to the root $r$ of the group $\mathcal{G}$, is at most $O(\ell)d(J_p, R \setminus J_p)$; and 
\item \label{dist_to_root2} the length of the path from $v\in J_p$ to the root $r'$ of its parent group $\mathcal{G'}$ (if it exists) is $O(\ell) d(J_p, R \setminus J_p)$.
\end{enumerate}
\end{lemma}

\begin{proof}
From the Observation \ref{obs:group_vol} any non singleton $\mathcal{G}$ has volume at most $2\ell$. From Claim~\ref{claim:clustering} each representative have volume at least $1-1/\ell$. Hence the total number of edges in the sub-tree corresponding to any group $\mathcal{G}$ is at most $O(\ell)$, which is property (\ref{group_edges1}).

Using Claim~\ref{claim:black_edges}, we can say that the length of any grey edge $e$ entering or leaving a black component $p$ is at least the length of any black edge belonging to $p$. And by the grouping procedure, we know that any grey edge entering a group will be at least as long as any grey edge within the group. These two observations together imply property (\ref{group_edges2}).

Lemma~\ref{lem:tau_properties} implies that the length of each edge within the black component $p$ is at most $d(J_p, R\setminus J_p)$, the length of the grey edge leaving $p$ is exactly $d(J_p, R \setminus J_p)$ and the length of the grey edges on the path from $p$ to the root are non-increasing, hence their total length is at most $d(J_p, R \setminus J_p)$. Also, by Claim~\ref{claim:black_edges}, we know that length of black edges corresponding to the root component of group $\mathcal{G}$ is at most the length of the grey edges entering it (in particular, the one on the path from $v$ to the root $r$) which is at most $d(J_p, R \setminus J_p)$ long. Combining it with property (\ref{group_edges1}) implies property (\ref{dist_to_root1}) for the path length from node $v$ to the root $r$. 

Now using a similar argument along with property (\ref{group_edges2}) for the grey edge entering $\mathcal{G'}$ on the path from $v$ to $r'$ and the fact that the total number of edges in the groups $\mathcal{G}$ and $\mathcal{G'}$ is at most $O(\ell)$, we can show property (\ref{dist_to_root2}) for the path length from node $v$ to the root $r'$.
\qed
\end{proof}

\begin{lemma}
\label{lem:big_movement}
Consider any representative $v \in R$. We can construct a new solution $\{x' , y', z'\}$ such that all the facilities from set $U_v$ are collocated with $v$. The cost of the new solution is at most $O(\ell)\mathbf{LP}$.
\end{lemma}

\begin{proof}
Consider a black component $p \in \Upsilon$, the set of representatives $J_p$ and the set of facilities $U_p = \cup_{v\in J_p} U_v$. Consider any representative $v \in J_p$ and the facilities $U_v$. We replace each $i \in U_v$ by $i_v$ which is collocated with $v$ and has the same opening as $i$: ${y'}_{i_v} := y_i$. Moreover, each client $j$ which sends some fraction of demand to $i$ now sends the same fraction to $i_v$, ${x'}_{j, i_v} : = x_{ji}$. For each $S$, we create the new set $S' = \{i_v~|~v \in J_p \wedge i \in S \cap U_v \}$, ${z'}_{S'}^{p}:={z}_S^{p}$. Moreover, for each facility $i$ and corresponding facility $i_v$ we have ${z'}_{\bot, i_v}^{p}:={z}_{\bot, i}^{p}$. The other assignment variables we define in a similar way. We modify the solution for each black component $p\in \Upsilon$ in the same way. Note that, since the old solution $\{x, y ,z\}$ was feasible w.r.t. the Standard LP and the Constraints (\ref{clp:6}) - (\ref{clp:13}) were satisfied for set $U_p$ for each black component $p\in \Upsilon$, hence the same holds for the new solution $\{x', y', z'\}$. Moreover, the cost of the new solution is at most $O(\ell)\mathbf{LP}$ by Lemma 4.1 in \cite{Li15uniform}.
\qed
\end{proof}

For any black component $p \in \Upsilon$, let $\dot{y}_p = \frac{\sum_{S \in \mathcal{S}}|S|{z'}_S^{p}}{1 - {z'}_{\bot}^{p}} = \frac{\sum_{i \in S, S \in \mathcal{S}}{z'}_{S, i}^{p}}{1 - {z'}_{\bot}^{p}}$ and ${y'}_p = \sum_{i \in U_p} {y'}_i =  \sum_{i \in S, S \in \tilde{\mathcal{S}}}{z'}_{S, i}^{p} = \sum_{S\in \mathcal{S}}|S| {z'}_{S,i}^{p} + \sum_{i \in {U_p}}{z'}_{\bot, i}^{p}$. Moreover, we define $\pi(J_p) :=  \sum_{j \in C} x_{p,j} (1-x_{p,j})$ for any $p \in \Upsilon$.

Next we show that, we can pre-process each black component $p$ by opening a set randomly from $\mathcal{S}$ and pre-assigning some clients to the open set. We send the demand of the rest of the clients that was served by $p$, to the root of the parent group. To do that, we first reduce the variance in the size of sets in $\mathcal{S}$. 

For the $|C|=ku$ case, we perform a massage process in which we move facilities from bigger sets to smaller ones, until the size of each set is either $\floor{\dot{y}_p}$ or $\ceil{\dot{y}_p}$. Using the saturation property, we can reroute the demand of clients assigned to these facilities, so that the final solution remains feasible. We scale up the opening values of these sets, so that the expected size of sets in $\mathcal{S}$ is $\dot{y}_p$. 

For the general case, instead, we use a brutal massage process in which we pick a prefix of the smallest sets in set $\mathcal{S}$, such that their total opening value is at least a constant. Then we add some extra facilities to the selected sets and scale up the opening values of these sets, so that the sets have size either $\floor{\dot{y}_p}$ or $\ceil{\dot{y}_p}$ and the total opening is exactly $\dot{y}_p$.

In both cases, we pick a set randomly and pre-assign some clients based on their connection values. The intuition is that the clients which are served by more than $1-\epsilon'$ by the black component $p$, get assigned to the selected set with high probability and the demand of the other clients can travel to the root of the parent group by paying the total cost of $O(\ell^2) \mathbf{LP}$.

\begin{lemma}
\label{lem:gen-pre-assign}
Let $p \in \Upsilon$ be a black component and $U_p$ satisfies $y_p \leq 2\ell$ and let $Z_p \in \{0 ,1\}$ be a random variable, such that $E[Z_p] = \dot{y}_p - \floor{\dot{y}_p}$. Moreover, constraints (\ref{clp:6}) to (\ref{clp:13}) are satisfied for the solution $\{x', y', z'\}$ and $U_p$. Then, we can pre-open a set $S \subseteq U_p$ of expected cardinality $\dot{y}_p$, where $|S| = \floor{\dot{y}_p} + Z_p$, and pre-assign a set $C' \subseteq C$ of clients to $S$ such that
	\begin{enumerate}
    	\item \label{gen-at_most_u} each facility $i \in S$ is pre-assigned at most $u$ clients
        \item \label{gen-out_demand} expected cost of sending not assigned demand $x_{p, C \setminus C'}$ to the root of the parent-group is at most $O(\ell^2) d(J_p, R \setminus J_p) \pi(J_p)$
        \item \label{gen-exp_y'_p} $\Pr[|S| = \floor{\dot{y}_p}]=\dot{y}_p - \floor{\dot{y}_p}$ and $\Pr[|S| = \ceil{\dot{y}_p}]=1 - (\dot{y}_p - \floor{\dot{y}_p})$
        \item \label{gen-cost_of_pre} expected cost of pre-assignment and local moving of $x_{p, C \setminus C'}$ is at most $O(\ell)\sum_{j \in C, i \in U_p} d(i, j){x'}_{i, j}$
	\end{enumerate}
\end{lemma}

The below proof of the Lemma \ref{lem:gen-pre-assign} is quite technical, it becomes somewhat easier in the saturated case, i.e., when $|C| = ku$ (see Appendix~\ref{sat_case}).

\begin{proof}
By splitting the sets and making copies (see Lemma 1 in \cite{spliting_svir}), we can assume that the solution $\{x', y', z'\}$ is complete w.r.t. any black component $p \in \Upsilon$ (see Definition \ref{complete_solution}). First we will describe a process of modifying a solution for black component $p \in \Upsilon$. Our goal is to decrease the variance in the size of sets with positive $z^p$ values and in parallel do not increase the cardinality of expected open set. A client $j$ is "decided" if ${x'}_{p,j} \geq 1 - \frac{1}{6\ell} = t$, and not decided otherwise. Let $C'' = \{j \in C~|~{x'}_{p,j} \geq t \}$ be the set of decided clients. Let us select $\mathcal{S}^1 \subseteq \mathcal{S}$, such that $\sum_{S \in \mathcal{S}^1} \frac{{z'}_S^p}{1-{z'}^{p}_{\bot}} = \frac{1}{3\ell}$. Also, we define ${z''}^p_{S}:=3\ell\frac{{z'}^p_{S}}{1-{z'}^{p}_{\bot}}$ and ${z''}^p_{S,i}:=3\ell\frac{{z'}^p_{S,i}}{1-{z'}^{p}_{\bot}}$ for all $S\in \mathcal{S}^1$ and $i\in S$. Moreover, ${z''}^p_{S, i, j}:=3\ell\frac{{z'}^p_{S,i,j}}{1-{z'}^{p}_{\bot}}$ for all $S\in \mathcal{S}^1, (i,j) \in S \cross C''$ and ${z''}_{\bot}^p = {z'}_{\bot}^p, {z''}_{\bot, i}^p = {z'}_{\bot,i}^p~\forall i \in U_p$. All the other variables ${z''}^p$ are equal zero. The variables $x'', y''$ we define such that constraints (\ref{clp:7}) and (\ref{clp:8}) are satisfied. All the demand of clients from $C \setminus C''$ which is served by $p$ we send to the root of the parent group of $G \ni p$.

Let us consider the set $S\in \mathcal{S}^1$ of maximum cardinality. Let $n_{\mathcal{S}^1}=|S|$. We know that $\dot{y}_p = \sum_{S\in \mathcal{S}} \frac{z^p_S|S|}{1-{z'}^{p}_{\bot}}$. The above equality gives 
$$
 \sum_{S \in \mathcal{S}^1} \frac{|S|{z'}^p_S}{1-z_{\bot}^p}  + \sum_{S \in \mathcal{S} \setminus \mathcal{S}^1} \frac{|S|{z'}^p_S}{1-z_{\bot}^p} = \dot{y}_p 
$$
Replacing $|S|=0$ for sets in $\mathcal{S}^1$ and $|S|=n_{\mathcal{S}^1}$ for sets in $\mathcal{S}\setminus \mathcal{S}^1$ we get, $\left(1- \frac{1}{3\ell} \right) n_{\mathcal{S}^1} \leq \dot{y}_p $. This implies 
$$n_{\mathcal{S}^1} \leq \frac{\dot{y}_p}{(1-\frac{1}{3\ell})} \leq \dot{y}_p + \frac{\dot{y}_p}{(3\ell-1)}  \leq \dot{y}_p + \frac{2\ell}{(3\ell-1)} < \dot{y}_p + 1$$
The last inequality assumes that $\ell$ is big enough, such that $\frac{2\ell}{(3\ell-1)}<1$. The above calculations imply that $ n_{\mathcal{S}^1} \leq \ceil{\dot{y}_p}$.

We have to increase the cardinality of each set from $\mathcal{S}^1$, such that it is either $\floor{\dot{y}_p}$ or $\ceil{\dot{y}_p}$ and $\sum_{S \in \mathcal{S}} {z''}_S^p |S|$ is exactly equal to $\dot{y}_p$. When we increase cardinality of some set $S$ by some number $d$ of facilities then we "create" $d$ facilities in the root of black component $p$ each with an opening ${z''}^p_S$ and we add this newly created facilities to the set $S$. For each $S \in \mathcal{S}^1$ which has cardinality smaller then $\floor{\dot{y}_p}$ we increase its cardinality up to $\floor{\dot{y}_p}$. Now, until $\sum_{S \in \mathcal{S}} {z''}_S^p |S|$ is smaller than $\dot{y}_p$ we pick any set $S \in \mathcal{S}^1 : |S| = \floor{\dot{y}_p}$ and increase by one its cardinality (It might happen that we need to replace $S$ with its two copies $S' = S''$, so that we do not increase $\sum_{S \in \mathcal{S}} {z''}_S^p |S|$ to much. They have positive openings ${z''}_{S'}^p, {z''}_{S''}^p$ which sum up to ${z''}_{S}^p$. We increase the cardinality of only one of these two sets). Finally, we have $\sum_{S \in \mathcal{S}} {z''}_S^p |S| = \dot{y}_p$ 
 
Observe that $\sum_{S \in \mathcal{S}} {z''}_{S}^{p} = 1$. Recall that the sets with positive ${z''}_S^p$ value have cardinality either $\floor{\dot{y}_p}$ or $\ceil{\dot{y}_p}$. Define $\mathcal{S}_1 = \{S \in \mathcal{S}~|~{z''}_S^p > 0 \wedge |S| = \floor{\dot{y}_p}\}$ and $\mathcal{S}_2 = \{S \in \mathcal{S}~|~ {z''}_S^p > 0 \wedge |S| = \ceil{\dot{y}_p}\}$. If variable $Z_p$ has value 0, then we select a set from $\mathcal{S}_1$, otherwise if $Z_p = 1$ then we should select it from $\mathcal{S}_2$. To do this we proceed in the following way. If $Z_p = 0$ then we set $\bar{z}_S^p = \frac{{z''}_S^p}{1 - (\dot{y}_p - \floor{\dot{y}_p})}$ and $\bar{z}_{S, i, j}^p = \frac{{z''}_{S, i, j}^p}{1 - (\dot{y}_p - \floor{\dot{y}_p})}$ for all $S \in \mathcal{S}_1, i\in S, j \in C''$ and all the other variable we set to zero. Otherwise, if $Z_p = 1$ then $\bar{z}_S^p = \frac{{z''}_S^p}{\dot{y}_p - \floor{\dot{y}_p}}$ and $\bar{z}_{S, i, j}^p = \frac{{z''}_{S, i, j}^p}{\dot{y}_p - \floor{\dot{y}_p}}$ for all $S \in \mathcal{S}_2, i\in S, j \in C''$ and all the other variables we set to zero. Note that the expected value of $\bar{z}_S^p (\bar{z}_{S, i ,j}^p)$ is equal to ${z''}_S^p ({z''}_{S, i, j}^p)$ for each $S \in \mathcal{S}$. Moreover, $\bar{z}^p$ is a probability distribution over sets from $\mathcal{S}_1 \cup \mathcal{S}_2$, because $\dot{y}_p - \floor{\dot{y}_p} = \sum_{S \in \mathcal{S}_2} {z''}_S^p$ and $1 - (\dot{y}_p - \floor{\dot{y}_p}) = \sum_{S \in \mathcal{S}_1} {z''}_S^p$. This implies that property (\ref{gen-exp_y'_p}) holds.

Now assume that we have selected a set $S \in \mathcal{S}_1 \cup \mathcal{S}_2$ according to this probability distribution. For each client $j \in C''$ and facility $i \in S$ let $w_{i,j} := \frac{{z''}^p_{S, i ,j}}{{z''}^p_{S, i}} = \frac{{z''}^p_{S, i ,j}}{{z''}_S^p} = \frac{z^p_{S, i ,j}}{z_S^p}$. Notice that $w_{i, j} \in [0,1]$, by constraint (\ref{clp:9}), $\sum_{i \in S} w_{i, j} \leq 1$ for every client $j \in C''$ by constraint (\ref{clp:11}), and $\sum_{j \in C''} w_{i,j} \leq u$ for every $i \in S$ by constraint (\ref{clp:12}). Since, the solution $\{x'', y'', z''\}$ is complete, hence  $w_{i, j} \in \{0,1\}$, for each client $j \in C''$ and facility $i \in S$. Thus, $w$ defines an integral matching between $C''$ and $S$ in which each client is matched at most once, and each facility from $S$ is matched at most $u$ times.

As we said before all the demand of clients from $C \setminus C''$ which is served by $p$ in LP solution we send to the root of the parent group of $\mathcal{G} \ni p$. The same we do for demand of clients from $C''$ which were not selected in the above procedure. For all the rest of the demand we assign it according to the output of the above procedure. From Lemma~\ref{lem:group_edges} we know that the distance to the root of the parent group is at most $O(\ell)d(J_p, R \setminus J_p)$.

For each $(i, j) \in S \cross C''$ by $j \rightarrow i$ we denote the event that client $j$ is pre-assigned to facility $i$.

$$\mathrm{Pr}[j \rightarrow i] = \mathrm{Pr}[Z_p = 0]  \sum_{S \in \mathcal{S}_1, i \in S} \bar{z}_S^p \frac{{z''}_{S, i, j}^p}{{z''}_S^p} + \mathrm{Pr}[Z_p = 1]  \sum_{S \in \mathcal{S}_2, i \in S} \bar{z}_S^p \frac{{z''}_{S, i, j}^p}{{z''}_S^p} =$$ $$\sum_{S \in \mathcal{S}_1 \cup \mathcal{S}_2, i \in S} {z''}_{S, i, j}^p \leq {x''}_{i,j}$$ 
Using the upper bound on $P[j \rightarrow i]$ we can bound the total pre-assignment cost in the following way.

$$\mathrm{E}[\text{total pre-assignment cost}] \leq \sum_{j \in C'', i \in U_p} d(i, j) 
\mathrm{Pr}[j \rightarrow i] \leq \sum_{j \in C'', i \in U_p} d(i, j) {x''}_{i, j}$$

For all the rest of the clients which were not assigned, we will send their demand to the root of the parent group via their respective representatives (for each representative $v$, facilities in $U_v$ are co-located with $v$). The cost of sending $x_{i,j}$ units of not pre-assigned demand to their representatives we can bound by $x_{i,j}d(i,j)$. Notice that decided client can be either pre-assigned or sent to the root, so in the worst case we can bound a total cost of pre-assignment and moving to representatives by

$$\sum_{j \in C'', i \in U_p} d(i, j){x''}_{i, j} + \sum_{j \in C, i \in U_p} d(i, j){x'}_{i, j} \leq O(\ell) \sum_{j \in C, i \in U_p} d(i, j){x'}_{i, j}$$

From the definition of decided clients, for each $j \in C''$, it follows ${x'}_{p, j}  \geq 1 - \frac{1}{6\ell}$. Moreover, $\sum_{S \in \mathcal{S}^1} {z'}_S^p = \frac{1 - {z'}_{\bot}^p}{3\ell}$, so we can lower bound ${x'}_{\mathcal{S}^1, j} = \sum_{S \in \mathcal{S}^1}\sum_{i \in S} {z'}^p_{S,i,j}$ in the following way.

$${x'}_{\mathcal{S}^1, j} \geq {x'}_{p, j} - (\sum_{S \in \tilde{\mathcal{S}}} {z'}_S^p - \sum_{S \in \mathcal{S}^1} {z'}_S^p) = {x'}_{p, j} - 1 + \frac{1 - {z'}_{\bot}^p}{3\ell}$$

For each $j \in C''$ we have, ${x''}_{p, j} = {x''}_{\mathcal{S}^1, j} = \frac{3\ell}{1 - {z'}_{\bot}^p} {x'}_{\mathcal{S}^1, j}$.

By $j \rightarrow$ we denote the event that client $j$ was assigned and $\Pr[j \rightarrow]$ the probability of this event. We will apply the following bound only for the decided clients $C''$.

$$\Pr[j\rightarrow] \geq {x''}_{p, j} = {x''}_{\mathcal{S}^1, j} \geq \frac{3\ell}{1-{z'}_{\bot}^p}({x'}_{p, j} - 1 + \frac{1 - {z'}_{\bot}^p}{3\ell}) = 1 + \frac{3\ell}{1-{z'}_{\bot}^p}({x'}_{p, j} - 1)$$
It might happen that a decided client $j$ is not assigned to his black component $p$. Then we send demand ${x'}_{p, j}$ to the root of his parent group. We upper bound the cost of this step in the following way.

$$\sum_{j \in C''} (1 - \mathrm{P}[j \rightarrow]) {x'}_{p, j} O(\ell) d(J_p, R \setminus J_p) 
\leq \sum_{j \in C''} (1 - {x''}_{p, j}) {x'}_{p, j} O(\ell) d(J_p, R \setminus J_p) \leq$$

\begin{eqnarray}
\label{gen-ineq_cost_out_1}
O(\ell^2) d(J_p, R \setminus J_p) \sum_{j \in C''} {x'}_{p, j}(1 - {x'}_{p, j}).&
\end{eqnarray}
where the inequality follows by the bound on $\Pr[j\rightarrow]$. Notice that decided client can be assigned in at most one black component. As we mentioned before, we will send all demand of not decided clients to the root. Hence, the cost for sending this demand is at most

\begin{eqnarray}
\label{gen-ineq_cost_out_2}
O(\ell) d(J_p, R \setminus J_p) \sum_{j \in C\setminus C''} {x'}_{p, j} &  \leq & O(\ell^2) d(J_p, R \setminus J_p) \sum_{j \in C\setminus C''} {x'}_{p, j} (1 - {x'}_{p, j})
\end{eqnarray}
where the inequality holds because ${x'}_{p, j} <1 - \frac{1}{6\ell}$.
By summing up inequalities (\ref{gen-ineq_cost_out_1}) and (\ref{gen-ineq_cost_out_2}) and by the fact that ${x'}_{p, j} = x_{p, j}$ we can bound the cost of sending not assigned demand from the representatives to the root by $O(\ell^2) d(J_p, R \setminus J_p) \pi(J_p)$.
\qed
\end{proof}

\subsection{Dependent Rounding}
\label{dp_sec}

We will use a dependent rounding (DR) procedure, described in \cite{srin_dep_round}, to decide if a particular variable should be rounded up or down. It transforms fractional vector $\{\bar{v}_i\}_{i = 1}^n$ to a random integral vector $\{\hat{v}_i\}_{i=1}^n$. DR procedure has the following properties:
\begin{enumerate}
\item \label{marginal} Marginal distribution: $Pr[\hat{v_i} = 1] = \bar{v}_i$
\item \label{sum_prev} Sum-preservation: $\sum_{i=1}^n \hat{v}_i \in \{\floor{\sum_{i=1}^n \bar{v}_i}, \ceil{\sum_{i=1}^n \bar{v}_i}\}$
\end{enumerate}
In our procedure we first fix a tree $\tau \in \Upsilon$. Then we choose a pair of fractional black components according to a predefined order. After that we increase the opening of one and decrease the opening of the other in a randomized way. After each such iteration, at least one black component has an integral opening. Based on the value of the integral opening $y{'''}_p \in \{\floor{\dot{y}_p}, \ceil{\dot{y}_p}\}$ decided for a black component $p \in \Upsilon$, we will select a set of facilities $S\in \mathcal{S}: |S|=y'''_p$ in a random way (for details see Lemma~\ref{lem:gen-pre-assign}). First we do dependent rounding among the black components of the children groups of each parent group. After this step each group $\mathcal{G}$ will have at most one fractional black component among all black components in its children groups, and the total opening (capacity) within these black components will be preserved. Finally, once we complete this rounding phase for all the trees in $\Upsilon$, then we will do dependent rounding among all the remaining fractional black components across all the trees in $\Upsilon$ in an arbitrary order. The procedure will preserve the sum of facility openings, hence in the end we will open exactly $k$ facilities.

In the "rounding among children groups" step, we will do the rounding among the black components within the children-groups of a group $\mathcal{G}$ in an order defined by non-decreasing distance of these black components to the root $r$ of the parent group $\mathcal{G}$ (breaking ties arbitrarily). This way, we would have an extra property on the number of open facilities for every prefix in this order of the black components belonging to the children-groups of group $\mathcal{G}$. 


Before we start the rounding procedure, we will send exactly $\sum_{i\in U_p} {z'}_{\bot,i}^p - \dot{y}_p {z'}_{\bot}^p$ opening, from each black component $p$, to a virtual black component $v_\mathcal{G}$ co-located with the root of the group. Note that since ${z'}_{\bot}^p = {z}_{\bot}^p$ and ${z'}_{\bot,i}^p = {z}_{\bot,i}^p$, we can use $z$ instead of $z'$. Let us define $\dot{y}_{v_{\mathcal{G}}} = \sum_{p \in \mathcal{G}} \sum_{i\in U_p} z_{\bot,i}^p - \dot{y}_p z_{\bot}^p$. We will call this the \emph{blue} opening. We will treat this blue opening, co-located with the root, as a virtual black component. Since we are in the uniform capacity case, by loosing a constant factor we can assume that $\mathcal{F}=\mathcal{C}$ \cite{Li15uniform}. Moreover we can work with soft capacitated version of the problem due to Theorem 1.2 in \cite{Li15uniform}. Hence, for the blue opening, we can simply open the decided number of co-located facilities at the virtual black component. By $BCG(\mathcal{G})$ we denote a set of all, virtual or not, black components in the children groups of group $\mathcal{G}$. Note that, from now on the group $\mathcal{G}$ also contains the virtual black component $v_{\mathcal{G}}$.

Consider the root group of the tree $\tau$. If it is a singleton root component then we classify it as a virtual black component, otherwise we treat it as a standard black component. Note that the sum $\dot{y}_p + \sum_{i \in U_p} {z'}_{\bot, i}^{p} - \dot{y}_p {z'}_{\bot}^{p} = \sum_{i \in U_p} z_{\bot, i}^{p} + (1 - z_{\bot}^p) \dot{y}_p = y_p = {y'}_p$. Hence, the total opening across all the black components is exactly equal to $k$.

\begin{lemma}
\label{lem:gen-blue}
For any group $\mathcal{G}$, the total demand is at most $(1+O(1/\ell)) \sum_{p \in \mathcal{G}} u \dot{y}_p$. 
\end{lemma}

\begin{proof}
Let us look at a black component $p \in \mathcal{G}$. Notice that the demand for any clients $j \in C \setminus C''$ remains $x'_{p, j}= x_{p, j}$(see Lemma~\ref{lem:gen-pre-assign}). For any clients $j' \in C''$, either they are pre-assigned, in which case their demand is $1$, otherwise their demand remains $x'_{p, j'} = x_{p, j'}$. Since, $x_{p, j'} \geq 1-\frac{1}{6\ell}$, for all $j' \in C''$, in the worst case the total demand associated with $p$ is at most
$$\sum_{j \in C'', i \in U_p} \frac{{x}^p_{i, j}}{1- 1/(6 \ell)} + \sum_{j \in C \setminus C'', i \in U_p} {x}^p_{i, j} \leq$$ $$\sum_{j \in C, i \in U_p} \frac{{x}^p_{i, j}}{1- 1/(6 \ell)} =  \frac{1}{1-1/(6\ell)} \left(\sum_{S \in \mathcal{S}, i \in S, j \in C} z_{S, i, j}^p +  \sum_{i \in U_p, j \in C} z_{\bot, i, j}^p\right)$$

First let us compare the opening value $\dot{y}_p z_{\bot}^{p}$ (which is sent to $v_{\mathcal{G}}$ by component $p$) to the original opening $\sum_{i \in U_p} z_{\bot, i}^{p}$, for some $p \in \mathcal{G} \setminus \{v_{\mathcal{G}}\}$.
\begin{eqnarray}
\label{ineq:gen-blue_opening}
\frac{\dot{y}_p z_{\bot}^{p}}{\sum_{i \in U_p} z_{\bot, i}^{p}} \leq \frac{2 \ell z_{\bot}^{p}}{\ell^2 z_{\bot}^{p}} = \frac{2}{l}
\end{eqnarray}
The inequality follows from the fact that $\dot{y}_p \leq 2\ell$ and from Constraint~(\ref{clp:13}).
By adding Constraint~(\ref{clp:12}) for set $\bot$ over all $i \in U_p$, we get $\sum_{i \in U_p}z_{\bot, i}^{p}u \geq \sum_{i \in U_p, j \in C}z_{\bot, i, j}^{p}$. Consider the demand (scaled up by a factor $\frac{1}{1-1/(6\ell)}$) served by the $\bot$ set in the initial LP solution and the blue opening sent from $p$ to $v_{\mathcal{G}}$ which is equal to $\sum_{i\in U_p} z_{\bot, i}^{p} - \dot{y}_p z_{\bot}^{p}$. Now we show that the ratio between this demand and the capacity corresponding to the blue opening is at most $(1+O(1/\ell))$. Note that the demand served by $\bot$ remains the same in all the modified solutions.
$$ \frac{1}{1-1/(6\ell)} \frac{\sum_{i \in U_p, j \in C}z_{\bot, i, j}^{p}}{u \left( \sum_{i \in U_p} z_{\bot, i}^{p} - \dot{y}_{p} z_{\bot}^p \right)} \stackrel{Const. (\ref{clp:12})} \leq \frac{1}{1-1/(6\ell)}\ \frac{u \sum_{i \in U_p} z_{\bot, i}}{u \left(\sum_{i \in U_p} z_{\bot, i}^{p} - \dot{y}_{p}z_{\bot}^p\right)} 
$$
$$
= \frac{1}{1-1/(6\ell)}\ \frac{1}{1-\frac{ \dot{y}_{p}z_{\bot}^p}{\sum_{i \in U_J} z_{\bot, i}^{p}}} 
\stackrel{Ineq. (\ref{ineq:gen-blue_opening})} \leq \frac{1}{(1-1/(6\ell))(1-2/\ell)} = (1+O(1/\ell))
$$

The following calculations show that the capacity violation on $v_{\mathcal{G}}$ w.r.t. to the total scaled up demand served by $\bot$ set of each $\{U_p : p \in \mathcal{G} \setminus \{v_{\mathcal{G}}\}\}$ and the blue opening sent from all the black components of that group is at most $(1+O(1/\ell))$.
\begin{eqnarray}
\label{ineq:gen-blue_opening_2}
\frac{1}{1-1/(6\ell)} \frac{\sum_{p \in \mathcal{G} \setminus \{v_{\mathcal{G}}\}} \sum_{i \in U_p, j \in C}z_{\bot, i, j}^{p}}{\sum_{p \in \mathcal{G} \setminus \{v_{\mathcal{G}}\}} u \left( \sum_{i \in U_p} z_{\bot, i}^{p} - \dot{y}_{p}z_{\bot}^p \right)} \leq
 \max_{p \in \mathcal{G} \setminus \{v_{\mathcal{G}}\}} \frac{6\ell}{6\ell-1} \frac{\sum_{i \in U_p, j \in C}z_{\bot, i, j}^{p}}{\sum_{i \in U_p} z_{\bot, i}^{p} - \dot{y}_{p}z_{\bot}^p} \leq 1+O(1/\ell)
\end{eqnarray}
 
Now we compare the rest of the demand $\frac{1}{1-1/(6\ell)} \sum_{S \in \mathcal{S}, i \in S, j \in C} z_{S, i, j}^p$ to the rest of the opening $\dot{y}_p$ for each black component $p \in \mathcal{G} \setminus \{v_{\mathcal{G}}\}$ and show that this demand is at most $u \dot{y}_p$. 
$$
 \sum_{S \in \mathcal{S}, i \in S, j \in C} \frac{z_{S, i, j}^p}{1-1/(6\ell)}  \stackrel{Cons. (\ref{clp:12})} \leq \sum_{S \in \mathcal{S}, i \in S} \frac{u {z}_{S, i}^p}{1-1/(6\ell)} \stackrel{Cons. (\ref{clp:10})}= \frac{1}{1-1/(6\ell)} \sum_{S \in \mathcal{S}} u |S| {z}_{S}^p
 $$
 $$
 = (1 + O(1/\ell))u \dot{y}_p (1-z_{\bot}^p) \leq  (1 + O(1/\ell))u \dot{y}_p
$$

Adding the above inequalities for each $p \in \mathcal{G} \setminus \{v_\mathcal{G}\}$ and the inequality (\ref{ineq:gen-blue_opening_2}) proves the lemma.
\qed
\end{proof}

For simplicity of exposition, we will say that \emph{a black component $p$ is closed}, if the procedure decides to round down the opening of that component to $\floor{\dot{y}_p}$, otherwise we say it is \emph{opened}. 

In this dependent rounding procedure, in contrast to \cite{BFRS15}, we will also be able to \emph{pull demand} to the black components where we decided to open an extra facility. Cost of pulling can be bounded by the LP cost for sending the demand out of a black component. This new strategy is crucial to bring down the capacity violation from $2+\epsilon$ to $1+\epsilon$.

\subsection{Rounding among children groups}
\label{rounding-children}

Consider any tree $\tau \in \Upsilon$ and its root $r$. For simplicity of description, we add a fake single node parent group and attach the root $r$ to this fake group node with a grey edge of length exactly $d(J_p, R\setminus J_p)$, where $p$ corresponds to the only the black component in the root group. Notice that from now on even the original root group is a child-group of some other group.

In the first phase of dependent rounding, we select the deepest (w.r.t. the number of edges) leaf-group and let its parent group be $\mathcal{G}$. Let $\bar{y}_p=\dot{y}_p - \floor{\dot{y}_p}$ for each $p \in \Upsilon$. For performing this dependent rounding procedure within children groups of $\mathcal{G}$, we use the root $r$ of $\mathcal{G}$ as a \emph{accumulator}, which will temporarily store all the not assigned demand from children groups of $\mathcal{G}$. Let $n_{\mathcal{G}} = |BCG(\mathcal{G})|$.

To perform the dependent rounding procedure, we would order the components in $BCG(\mathcal{G}) = \{p_1, p_2, \dots p_{n_{\mathcal{G}}}\}$ by non-decreasing distance from the root $r$ of $\mathcal{G}$, so $d(p_i, r) \leq d(p_{i+1}, r)$ for $i < n_{\mathcal{G}}$. We define the vectors  $\dot{y}_{\mathcal{G}} = (\dot{y}_{p_1}, \dot{y}_{p_2}, \dots \dot{y}_{p_{n_{\mathcal{G}}}})$ and $\bar{y}_{\mathcal{G}} = (\bar{y}_{p_1}, \bar{y}_{p_2}, \dots \bar{y}_{p_{n_{\mathcal{G}}}})$. Now we apply dependent rounding between the two fractional components in the $i$th prefix of vector $\bar{y}_{\mathcal{G}}$, for each $i$ starting from $i=2$ until $i=n_{\mathcal{G}}$. Note that, after applying dependent rounding on the $i$th prefix of $\bar{y}_{\mathcal{G}}$, at most one component will remain fractional in the prefix and one will become integral. If the black component $p$ which become integral is not virtual, we apply Lemma~\ref{lem:gen-pre-assign}, with $r \in \mathcal{G}$ as a root and, with $Z_p = 1$ if component is open and $Z_p = 0$ if it is closed. Let the output vector be $Z_{\mathcal{G}}=(Z_{p_1}, Z_{p_1}, \dots Z_{p_{\mathcal{G}}})$. If $\sum_{i} \bar{y}_{\mathcal{G}}(i)$ is not integral then the output vector will have one fractional variable, otherwise it will be a vector of all integral values. Notice that by the property (\ref{marginal}) of dependent rounding $E[Z_p] = \bar{y}_p$ and by the property (\ref{sum_prev}) the sum of the facility opening $\dot{y}_p$ is preserved.

From now on, we will ignore the presence of all the children of the group $\mathcal{G}$ in our procedure. We repeat this process until our tree $\tau$ has only the added fake group left. Note that the root group will contain at most one black component which is fractional. After we finish the first phase, for each group $\mathcal{G}$, at most one component of $Z_{\mathcal{G}}$ will be fractional.

For any vector $v$, let $v[i_1,i_2] = \sum_{k=i_1}^{i_2} v(k)$. Due to the ordering which we follow in the above dependent rounding procedure and the fact that we didn't move any opening out of (or into) set $BCG(\mathcal{G})$ for each group $\mathcal{G}$, the following observation holds.

\begin{observation}
\label{obs:prefix_opening}
After the first phase of the rounding procedure, $Z_{\mathcal{G}}[1,i] \in \left[ \floor{\bar{y}_{\mathcal{G}}[1,i]}, \ceil{\bar{y}_{\mathcal{G}}[1,i]} \right]$ holds for each $i$ and, for each non-leaf group $\mathcal{G}$. Moreover, $Z_{\mathcal{G}}[1,n_{\mathcal{G}}] = \bar{y}_{\mathcal{G}}[1,n_{\mathcal{G}}]$.
\end{observation}

Once we complete phase one of rounding for each tree $\tau \in \Upsilon$, the second phase of the rounding procedure starts. In the second phase of the rounding procedure we just apply dependent rounding among all the remaining fractional variables, in an arbitrary order, until everything is integral. We apply Lemma \ref{lem:gen-pre-assign} to all the non-virtual components with $Z_p = 1$, if it was open, and $Z_p = 0$ otherwise. Notice that for the black component from the root group we will use the root of a fake group as an accumulator. We open $\ceil{\dot{y}_{v_{\mathcal{G}}}}$ facilities in each virtual component $v_{\mathcal{G}}$ if it was rounded up, and $\floor{\dot{y}_{v_{\mathcal{G}}}}$ otherwise.

Since the last fractional component in $BCC(\mathcal{G})$ could be either opened or closed, the total $Z_{\mathcal{G}}[1,n_{\mathcal{G}}]$ is either $\floor{\bar{y}_{\mathcal{G}}[1,n_{\mathcal{G}}]}$ or $\ceil{\bar{y}_{\mathcal{G}}[1,n_{\mathcal{G}}]}$ respectively. And since each of the components of vector $Z$ is integral, the following observation is true.

\begin{observation}
\label{obs:prefix_opening_2}
After the second phase of the rounding procedure, $Z_{\mathcal{G}}[1,i] \in$ $	\left\{ \floor{\bar{y}_{\mathcal{G}}[1,i]}, \ceil{\bar{y}_{\mathcal{G}}[1,i]}\right\}$ holds for each $i$ and, for each non-leaf group $\mathcal{G}$.\end{observation}

\begin{lemma}
\label{lem:flows_cost}
The cost of moving the demand from all black components to their respective accumulators can be bounded by $\sum_{p \in \Upsilon} O(\ell^2) d(J_p, R \setminus J_p) \pi(J_p)$. 
\end{lemma}

\begin{proof}
Notice that in Lemma \ref{lem:gen-pre-assign} we move all not pre-assigned demand from the black components to their respective accumulators. For any black component $p$ the cost of this operation is bounded by $O(\ell^2) d(J_p, R \setminus J_p) \pi(J_p)$. Note that, this bound is also valid for the black component $p$ corresponding to the root group, since the single node fake group is precisely at a distance $d(J_p, R \setminus J_p)$ from it. If we take the sum over all the black components in our instance we get the lemma.
\qed
\end{proof}

\subsection{Pulling back demand to the open facilities}
\label{pulling_back_sec}

Now we will define a single-commodity flow corresponding to distributing the demand from the accumulator co-located with the root of some non-leaf group to the open facilities in its children groups, for each tree $\tau \in \Upsilon$. To do this, we will pull back demand to the black components in a greedy way by pulling the demand first to the component belonging to $BCG(\mathcal{G})$ which is closest to the root $r$. We can bound the cost of pulling demand to the open facilities by charging it to the cost of pushing the demand to the root bounded in Lemma \ref{lem:flows_cost}. The intuition is, since we are pulling back the demand in a greedy fashion, we can argue that for every demand, the distance which it will travel in the pulling phase is at most the distance it traveled to reach the accumulator $r$ in the pushing phase. Since the cost for pushing is bounded by $O(\ell^2) \sum_{p \in BCG(\mathcal{G})} d(J_p, R \setminus J_p)\pi(J_p)$ (see Lemma~\ref{lem:flows_cost}), hence by the above claim the cost of pulling back the demand is bounded as well.

In this procedure, we first fix a tree $\tau \in \Upsilon$. Consider a non-leaf group $\mathcal{G}$ of $\tau$ and the set $BCG(\mathcal{G})$. In the pre-assignment step (Lemma \ref{lem:gen-pre-assign}), let $q_p$ be the amount of demand we assigned to the open facilities in each black component $p \in BCG(\mathcal{G})$. Notice that for any virtual black component $p \in BCG(\mathcal{G})$ we didn't assign any demand in a pre-assignment, so $q_p = 0$. Now we define vector $q_{\mathcal{G}} = (q_{p_1}, q_{p_1} \dots, q_{p_{n_{\mathcal{G}}}})$, to be the vector of the pre-assigned demand, which respects the same order of the components as in vector $\dot{y}_{\mathcal{G}}$.

Now we describe the pulling back procedure which we call \emph{the greedy pulling process}. First, we freeze $(1+O(1/\ell))u$ units of demand at the accumulator $r$ of group $\mathcal{G}$. Next we start pulling the rest of the demand to the black components $BCG(\mathcal{G})$. We do the pulling process in the same greedy order in which we did the dependent rounding among the black components $BCG(\mathcal{G})$, i.e. starting from the component closest to $r$. By definition, the vectors $\dot{y}_{\mathcal{G}}$, $Z_{\mathcal{G}}$ and $q_{\mathcal{G}}$ respect this ordering. We start pulling the demand equal to $(1+O(1/\ell))(Z_{\mathcal{G}}(i) + \floor{\dot{y}_{\mathcal{G}}(i)})u - q_{\mathcal{G}}(i)$ from the accumulator $r$ to the $i$th component starting from $i=1$, until we have no more demand to pull. We do this process for each non-leaf group $\mathcal{G}$ in all the trees in our forest $\Upsilon$.

\begin{observation}
\label{obs:cap_violation}
After the greedy pulling process, each black component $p \in BCG(\mathcal{G})$ has a capacity violation by a factor of at most $(1+O(1/\ell))$.
\end{observation}

\begin{lemma}
\label{lem:acc_leftover}
After the greedy pulling procedure, the left over demand at any accumulator $r$ of some non-leaf group $\mathcal{G}$ is exactly equal to $u(1+O(1/\ell))$; which is the demand frozen at the beginning.
\end{lemma}

\begin{proof}
By Lemma \ref{lem:gen-blue}, the total demand sent to the accumulator $r$ by groups in $BCG(\mathcal{G})$ is at most $(1+O(1/\ell)) u \left( \dot{y}_{\mathcal{G}}[1,n_{\mathcal{G}}]  \right) -  q_{\mathcal{G}}[1,n_{\mathcal{G}}]$ units. The total capacity of open facilities in $BCG(\mathcal{G})$ is 
$$u \left(Z_{\mathcal{G}}[1,n_{\mathcal{G}}] + \floor{\dot{y}_{\mathcal{G}}[1,n_{\mathcal{G}}]} \right) \stackrel{Obs.~\ref{obs:prefix_opening_2}}\geq u \left( \floor{\bar{y}_{\mathcal{G}}[1,n_{\mathcal{G}}]} + \floor{\dot{y}_{\mathcal{G}}[1,n_{\mathcal{G}}]} \right) $$
$$ \geq u \left( \bar{y}_{\mathcal{G}}[1,n_{\mathcal{G}}] - 1 + \floor{\dot{y}_{\mathcal{G}}[1,n_{\mathcal{G}}]} \right) = u \left( \dot{y}_{\mathcal{G}}[1,n_{\mathcal{G}}]  - 1 \right)$$

Now the total residual capacity of open facilities in $BCG(\mathcal{G})$, that can be used for pulling demand from the accumulator, after removing the pre-assigned demand $q_{\mathcal{G}}$ and taking into account the capacity violation of $(1+O(1/\ell))$ is at least $(1+O(1/\ell)) u \left( \dot{y}_{\mathcal{G}}[1,n_{\mathcal{G}}]  - 1\right) -  q_{\mathcal{G}}[1,n_{\mathcal{G}}]$. This quantity is precisely $u(1+O(1/\ell))$ units less than the upper bound on the demand located at the accumulator $r$, hence the lemma follows.
\qed
\end{proof}

\begin{lemma}
\label{lem:gen-pull_cost}
For any non-leaf group $\mathcal{G}$, the distance travelled by any demand in the greedy pulling phase is at most the distance travelled by it in the dependent rounding phase.
\end{lemma}

\begin{proof}
Let us order the demand located at the accumulator $r$ in non-decreasing order by the distance it travelled from the black component $p$, from which this demand originated, to the accumulator $r$. In the greedy pulling process, we would first reserve the closest $u(1+O(1/\ell))$ units of demand with the accumulator $r$. Then in the pulling process, we would pull the remaining demand in the order defined above, hence, the closest black component to $r$ would pull the demand coming first in this order. Now we would prove the lemma by induction. Note that the base case is true, since the initial $u(1+O(1/\ell))$ units had to travel zero distance in the greedy pulling phase and the demand that the $1$st component start pulling could only come from the same or some farther component. Now let us assume that the lemma holds until we start pulling for the $i$th component. The order of the demand imply that if the claim is true for the first demand which some $i$th component pulls in this procedure, then the claim would be true for the rest of the demand pulled by $i$. Hence, we just need to argue that the claim is true for the first demand which the $i$th component pulls.  Let $\epsilon>0$ be the smallest amount of demand, which is to be pulled next in the sequence, such that all this demand was located with the same black component $p'$ before it was pushed to the accumulator $r$. Since the claim holds till this point, it means that the demand last pulled by $(i-1)$th component was pushed by some black components coming after the $(i-2)$th component, in the sequence defined by the vector $\dot{y}_{\mathcal{G}}$. By the ordering of the demands, it follows that $p'$ also lies after $(i-2)$th component in the sequence. If $p'$ is different from the $(i-1)$th component, then we are done. Otherwise, $p'$ could be precisely $(i-1)$th component. Now we would show that after we finished pulling demand for the $(i-1)$th component, we would exhaust all the demand pushed by the components until the $(i-1)$th component in the sequence. This would imply that the component $p'$, from which the next demand $\epsilon>0$ was pushed, could not be the $(i-1)$th component in the sequence and this would complete the proof by induction. 

Let us consider a child-group $\mathcal{G}_1$ of $\mathcal{G}$. Using Lemma \ref{lem:gen-blue}, the total demand sent to the accumulator $r$ by the components in $\mathcal{G}_1$ is at most $(1+O(1/\ell)) u \sum_{p \in \mathcal{G}_1}\dot{y}_{p} -  \sum_{p \in \mathcal{G}_1} q_{p}$ units. For the distance analysis, we can assume that among the demand coming from a child-group $\mathcal{G}_1$, exactly $(1+O(1/\ell)) u \dot{y}_{v_{\mathcal{G}_1}}$ units of demand is coming from the virtual component corresponding to the children-group $\mathcal{G}_1$. Since the virtual component is the closed component to the accumulator $r$ among all the black components in $\mathcal{G}_1$, assuming this can only lower the distance over which the pushed demand would have travelled. If we can show that the claim holds for these distances, it will also hold for the original push distances. Using this modification and an argument similar to the proof for Lemma~\ref{lem:gen-blue}, we have that the demand pushed by components starting from $1$ until $(i-1)$th is at most $(1+O(1/\ell)) \left( u \dot{y}_{\mathcal{G}}[1, i-1] \right) - q_{\mathcal{G}}[1, i-1]$, where $q_{v_{\mathcal{G}_1}}=0$, for any virtual component corresponding to some child-group $\mathcal{G}_1$. The demand pulled by the same set of components is

$$ 
(1+O(1/\ell)) u \left(Z_{\mathcal{G}}[1,i-1] + \floor{\dot{y}_{\mathcal{G}}[1,i-1]} \right) - q_{\mathcal{G}}[1, i-1]  \stackrel{Obs.~\ref{obs:prefix_opening}}\geq $$

$$(1+O(1/\ell)) u \left( \floor{\bar{y}_{\mathcal{G}}[1, i-1]} + \floor{\dot{y}_{\mathcal{G}}[1,i-1]} \right) - q_{\mathcal{G}}[1, i-1]
$$

$$ \geq (1+O(1/\ell)) u \left( \bar{y}_{\mathcal{G}}[1,i-1] - 1 + \floor{\dot{y}_{\mathcal{G}}[1, i-1]} \right) - q_{\mathcal{G}}[1, i-1] \geq$$

$$ (1+O(1/\ell)) u \left( \dot{y}_{\mathcal{G}}[1,i-1]  - 1\right) - q_{\mathcal{G}}[1, i-1] $$

Hence the total demand pulled by components from $1$ until $i-1$ plus the reserved $u(1+O(1/\ell))$ units of demand is at least the demand pushed by components from $1$ until $i-1$. Hence, all this demand will be exhausted after we finished pulling demand for the $(i-1)$th component and the lemma follows.
\qed
\end{proof}

Now we would distribute the demand received by any black component $p$ to the actual open facilities (which are located at the representatives $J_p$), such that each facility has a capacity violation of at most $1+O(1/\ell)$. The following lemma bounds the cost of this step. 
\begin{lemma}
\label{lem:back_to_facilities}
Any demand that a black component $p$ received in the greedy pulling back process can be distributed to the open facilities within $p$. The distance travelled by the demand received by $p$ in this procedure is at most  $O(\ell)d(J_p, R \setminus J_p)$.
\end{lemma}

\begin{proof}
Note that, the grey edge going out of any black component $p \in \Upsilon$ (including the root component) is longer than any black edge within $p$ and the number of edges within any black component is $O(\ell)$. By Lemma~\ref{lem:tau_properties}, the length of this grey edge is $d(J_p, R \setminus J_p)$. Also note that, while traveling to the accumulator in the dependent rounding phase, all the demand definitely use this grey edge. Hence we can distribute all the demand received by $p$ in the greedy pulling phase to the exact location of the open facilities within $p$, by loosing an additive factor of $O(\ell) d(J_p, R \setminus J_p)$ in the distance travelled by the demand.
\qed
\end{proof}

\subsection{Distributing frozen demand to the open facilities}
\label{frozen_dist}

Now, we distribute the frozen $(1+O(1/\ell))u$ units of demand located at the accumulators over some open facilities, such that each open facility gets at most $uO(1/\ell)$ more demand. Let us fix a tree $\tau \in \Upsilon$. To do this distribution, we first send $(1+O(1/\ell))u$ units of demand from each of the non-fake accumulator to the accumulator of his parent group. Note that, using Lemma~\ref{lem:group_edges}, we can bound the cost for this movement by paying an additive factor of $O(\ell) d(J_p, R\setminus J_p)$ in the distance moved by this demand in Section \ref{rounding-children}. Let $r$ be the accumulator belonging to the group $\mathcal{G}$, which received $|C_{\mathcal{G}}| (1+O(1/\ell))u = O(u|C_{\mathcal{G}}|)$ units of demand from the accumulators of the non-leaf children groups $C_{\mathcal{G}}$ of the  group $\mathcal{G}$ in the tree $\tau$. Note that, $|C_{\mathcal{G}}| \leq n_{\mathcal{G}}$, since $\mathcal{G}$ may have children which are leaf-groups. We start sending $O(1/\ell) u (\floor{\dot{y}_\mathcal{G}(i)} + Z_\mathcal{G}(i))$ units of demand to the $i$th black component (in the same greedy order defined by the vector $\dot{y}_{\mathcal{G}}$) in the $BCG(\mathcal{G})$, starting from $i=1$, until we have no more demand left with $r$. 

\begin{lemma}
\label{lem:acc_distribution}
After the distribution procedure for some accumulator $r$ belonging to the group $\mathcal{G}$, all the demand which $r$ received from the accumulators of his non-leaf children groups will be distributed fully.
\end{lemma}

\begin{proof}
Note that any of the non-leaf group in $C_{\mathcal{G}}$ had $O(\ell)$ opening before the dependent rounding procedure. Now, by Observation~\ref{obs:prefix_opening_2}, we know that $Z_{\mathcal{G}}[1,n_{\mathcal{G}}]$ is at least $\floor{\bar{y}_{\mathcal{G}}[1,n_{\mathcal{G}}]}$. Hence, the total demand which we distribute over the black components in $BCG(\mathcal{G})$ in the above process is at least
$$
O(1/\ell) u (\floor{\dot{y}_\mathcal{G}[1,n_{\mathcal{G}}]} + Z_\mathcal{G}[1,n_{\mathcal{G}}]) \geq O(1/\ell) u (\floor{\dot{y}_\mathcal{G}[1,n_{\mathcal{G}}]} + \floor{\bar{y}_{\mathcal{G}}[1,n_{\mathcal{G}}]}) 
$$
$$ \geq O(1/\ell) u (\dot{y}_\mathcal{G}[1,n_{\mathcal{G}}] - 1) \geq O(1/\ell) u (|C_{\mathcal{G}}| \ell - 1) = O(u|C_{\mathcal{G}}|)
$$

Hence the lemma follows.
\qed
\end{proof}

By an argument similar to Lemma~\ref{lem:back_to_facilities}, we can send this demand to any open facility within $p$, by loosing an additive factor of $O(\ell) d(J_p, R \setminus J_p)$ in the distance traveled by the demand. Hence, this shows that we can distribute all the demand received by $r$ from his children-accumulators, corresponding to groups $C_{\mathcal{G}}$, among the open facilities within black components in $BCG(\mathcal{G})$, such that each facility receives at most an extra $O(1/\ell)u$ units of demand. We keep on doing this process bottoms-up, until we reach the very root fake accumulator. Now, for the demand located in the fake accumulator, we just distribute that demand over the open facilities in the very root group of the tree, which was using this accumulator. Note that since the very root group comprises of only one black component with $O(\ell)$ opening, there will be at least $O(\ell)$ open facilities in this component and again we will send $O(1/\ell) u$ units of extra demand to all the open facilities in the very root black component. By Lemma~\ref{lem:tau_properties}, each edge in the black component $p$ has length at most $d(J_p, R \setminus J_p)$ and by Lemma~\ref{lem:group_edges} the number of edges in the group is $O(\ell)$. Hence, by loosing an additive factor of $O(\ell) d(J_p, R \setminus J_p)$ in the distance traveled by the demand, we can distribute this demand over open facilities in this black component. 

In the following lemma, we bound the cost of distributing the frozen demand by the upper bound which we use to bound the cost of moving this demand from black components in $BCG(\mathcal{G})$ to the accumulator.

\begin{lemma}
\label{lem:acc_cost}
The distance travelled by any demand from each non-fake accumulator group $\mathcal{G}$ in the above re-distribution process is bounded above by $O(\ell) d(J_p, R \setminus J_p)$, which is also a bound on the distance it travelled to reach the accumulator in the dependent rounding phase.
\end{lemma}

\begin{proof}
Let us pick the farthest black component in each non-leaf group $\mathcal{G}' \in C_{\mathcal{G}}$, say $p_{\mathcal{G}'}$, to the accumulator root $r$ of group $\mathcal{G}$. Let set $A_{\mathcal{G}}=\{p_{\mathcal{G}'} |\mathcal{G}' \in C_{\mathcal{G}}\}$. Now pick the black component $p_1$, such that it is closest to the accumulator root $r$ among the components in set $A_{\mathcal{G}}$. Now we would argue that there are at least $O(\ell)$ facilities open in the components until $p_1$ in the sequence defined by the vector $\dot{y}_{\mathcal{G}}$. Let $\mathcal{G}_1 \ni p_1$. By definition of $\dot{y}_{\mathcal{G}}$, all the other black components of $\mathcal{G}_1$ comes before $p_1$ in the sequence. Let, $i_{p_1}$ be the position of $p_1$ in the sequence defined by vector $\dot{y}_{\mathcal{G}}$. By Observation~\ref{obs:prefix_opening_2} and the fact that $\sum_{p \in \mathcal{G}}\dot{y}_p = \sum_{p \in \mathcal{G}}{y}_p = O(\ell)$, the number of facilities open until $p_1$ is at least
$$
\floor{\dot{y}_\mathcal{G}[1,i_{p_1}]} + Z_\mathcal{G}[1,i_{p_1}] \geq \floor{\dot{y}_\mathcal{G}[1,i_{p_1}]} + \floor{\bar{y}_{\mathcal{G}}[1,i_{p_1}]} \geq \dot{y}_\mathcal{G}[1,i_{p_1}] - 1 \geq O(\ell)
$$
Now, look at the $u (1+O(\ell))$ units of demand coming from the accumulator belonging to group $\mathcal{G}_1$. We can think of sending this demand first from the accumulator $r$ to the open facilities in $BCG(\mathcal{G})$. By the above claim, all this demand would be sent to components at a distance at most $d(r, J_{p_1})$ in the sequence. Now, since all this demand had travelled from the children groups of $\mathcal{G}_1$ to its accumulator root, they must have travelled a distance at least the length of the shortest grey edge, lets say $e$, entering group $\mathcal{G}_1$. By Lemma~\ref{lem:tau_properties} and Lemma~\ref{lem:group_edges}, we know that the total number of edges in group $\mathcal{G}_1$ and $\mathcal{G}$ is at most $O(\ell)$ and each edge in group $\mathcal{G}_1$ and consequently in group $\mathcal{G}$ has length at most the length of $e$. This means, that the distance travelled by the demand reserved at the accumulator of group $\mathcal{G}_1$ travels at most $O(\ell) d(J_p, R \setminus J_p)$, which proves the lemma for this demand. 

Now we can ignore this $u (1+O(\ell))$ units of demand and the facilities to which this demand was distributed in the above process. Now, the same argument can be applied iteratively for the black component $p_2$ that is the next closest component to the accumulator root $r$ among the components in set $A_{\mathcal{G}}$ and the demand coming from the accumulator corresponding to the group $\mathcal{G}_2$ to which $p_2$ belongs. This proves the lemma.
\qed
\end{proof}

\begin{proof} [of the Theorem \ref{thm:general}]

We modify the initial solution by "moving" all facilities to their respective representatives (see Lemma~\ref{lem:big_movement}). The obtained solution has cost $O(\ell)\mathbf{LP}$. In the Lemma~\ref{lem:gen-pre-assign}, we pre-assign some demand and all the other demand we send to the respective accumulators. The cost of this operation is bounded by $\sum_{p \in \Upsilon} \sum_{j \in C, i \in U_p} O(\ell) d(i,j){x'}_{i,j} + O(\ell^2) \sum_{p \in \Upsilon} d(J_p, R \setminus J_p) \pi(J_p) \leq O(\ell^2)\mathbf{LP}$. The last inequality follows from \cite{Li16non_uniform}. By  the Lemmas~\ref{lem:gen-pull_cost}, \ref{lem:back_to_facilities} and \ref{lem:acc_cost}, we can bound the distance travelled by any demand in Sections~\ref{pulling_back_sec}, \ref{frozen_dist} by the distance it travelled in Section~\ref{rounding-children}. This implies that the cost of moving the demand in Sections \ref{pulling_back_sec} and \ref{frozen_dist} is bounded by $O(\ell^2)\mathbf{LP}$. Hence, overall the connection cost of our algorithm is $O(\ell^2)\mathbf{LP}$.

From the Observation \ref{obs:cap_violation} we know that the capacity violation of each facility is at most $1 + O(1/\ell)$. Moreover in Section~\ref{frozen_dist}, we increase the capacity violation of each facility by at most $O(1/\ell)$. So the final capacity violation is $1 + O(1/\ell)$, which ends the proof of the theorem.
\qed
\end{proof}

\section{Concluding remarks}

We showed that Configuration LP helps obtaining an algorithm with $1+\epsilon$ capacity violation for uniform capacities. It remains open if a similar result is possible for general capacities. It seems that the difficulty of generalizing our algorithm to general case lies in the dependent rounding. It is hard to control the number of open facilities and the capacities at the same time.

\bibliographystyle{abbrv}
\bibliography{kMedianbib}

\newpage

\appendix

\section{Rounding Algorithm (the $|C|=ku$ case)}
\label{sat_case}

The idea of the algorithm is roughly the same as in the general case. We just need to replace the Lemmas \ref{lem:gen-pre-assign} and \ref{lem:gen-blue} by their versions, which works only in fully saturation case, the Lemmas \ref{lem:massage} - \ref {lem:blue}.

Here we start with the solution $\{x', y', z'\}$. In the following lemmas we will first massage each black component. Next we will open the set of cardinality either $\floor{\dot{y}_p}$ or $\ceil{\dot{y}_p}$ and assign some demand to the open facilities from this set. All the not assigned demand we will send to the root of the parent group.

\begin{definition}
A black component $p \in \Upsilon$, with set of facilities $U_p$, such that for each $S \in \mathcal{S}$ if ${z}_S^{p}>0$ then $|S| \in \{\floor{\dot{y}_p}, \ceil{\dot{y}_p}\}$ is called a massaged black component w.r.t. the solution $\{x, y, z\}$.
\end{definition}

\begin{lemma} \label{lem:massage}
Given an LP solution $\{\mathbf{x'}, \mathbf{y'}, \mathbf{z'}\}$, we can get another solution $\{\mathbf{x''}, \mathbf{y''}, \mathbf{z''}\}$, such that:

\begin{enumerate}

\item All the black components, except the singleton root components, are massaged w.r.t the new solution $\{\mathbf{x''}, \mathbf{y''}, \mathbf{z''}\}$.
\item The cost of modifying the solution $\{\mathbf{x'}, \mathbf{y'}, \mathbf{z'}\}$ to the new solution $\{\mathbf{x''}, \mathbf{y''}, \mathbf{z''}\}$ is at most $O(\ell)\mathbf{LP}$.
\end{enumerate}
\end{lemma}

\begin{proof}
Now we will describe the process for modifying the LP solution facility by facility. We start with $\{\mathbf{x''}, \mathbf{y''}, \mathbf{z''}\} = \{\mathbf{x'}, \mathbf{y'}, \mathbf{z'}\}$. Pick $S_1, S_2 \in \{S \in \mathcal{S}~|~ {z'}_{S}^p > 0\}$ which are the biggest and the smallest sets respectively. Notice that $|S_1| \leq {\ell}_1$. If $|S_1| \leq |S_2| + 1$, then we are done. Else, let $i_1 \in S_1 \setminus S_2$. We distinguish the following two cases.

\paragraph{Case 1:} If for all clients $j\in C$, for which ${z''}_{S_1, i_1, j}^{p} = {z'}_{S_1, i_1}^{p}$, we also have ${z''}_{S_2, i_1, j}^{p} = 0$ then we call $i_1$ "clean". Now we can directly move $i_1$ to $S_2$ without incurring any cost. More formally, we add a copy of sets $S'_1:=S_1 \setminus \{i_1\}$ and  $S'_2:=S_2 \cup \{i_1\}$ to $\mathcal{S}$ and remove a copy of sets $S_1$ and $S_2$ from it. Due to the splitting and making copies step above, $\mathcal{S}$ is a multi-set, hence a set may have several copies in $\mathcal{S}$. Here we assume that $S_1$, $S_2$, $S'_1$ and $S'_2$ refers to some copy of these sets. We modify the solution $z''$ as follows:
\begin{enumerate}
\item Facility movement:  
\begin{description}
\item[-] ${z''}_{S'_1}^{p} = {z''}^{p}_{S_1}$, ${z'}_{S'_2}^{p} = {z'}^{p}_{S_2}$, 
\item[-] ${z''}^{p}_{S_1}= {z''}^{p}_{S_2}= 0$,
\item[-] ${z''}_{S'_1 ,i}^{p} = {z''}_{S'_1}^{p} \, \forall i \in S'_1$, 
\item[-] ${z''}_{S'_2, i}^{p} = {z''}_{S'_2}^{p} \, \forall i \in S'_2$,
\item[-] ${z''}_{S_1, i}^{p} = {z''}_{S_2, i}^{p} = 0 \, \forall i $.
\end{description}
\item Clients reassignment:  
\begin{description}
\item[-] ${z''}_{S'_2, i, j}^{p} =  {z''}_{S_2, i, j}^{p} \, \forall j \in C, i \in S_2$,
\item[-] ${z''}_{S'_2, i_1, j}^{p} =  {z''}_{S_1, i_1, j}^{p} \, \forall j \in C$,
\item[-] ${z''}_{S'_1, i, j}^{p} =  {z''}_{S_1, i, j}^{p} \, \forall j \in C, i \in S'_1$,
\item[-] ${z''}_{S_1, i, j}^{p} = {z''}_{S_2, i, j}^{p} = 0 \, \forall i,j $
\end{description}

\end{enumerate}
\paragraph{Case 2:} If there exists client $j \in C$ for which ${z''}_{S_1, i_1, j}^{p} = {z'}_{S_1, i_1}^{p}$ and there exists facility $i_2 \in S_2$ such that ${z''}_{S_2, i_2, j}^{p} = {z'}_{S_2, i_2}^{p}$, then we call $i_1$ "dirty". Observe that because we are in the saturated case, we know that the set $S_1$ serves at least $2u{z'}_{S_1}^{p}$ more demand than the set $S_2$. 
We also know that each client is sending either $0$ or ${z'}_{S_1}^{p}$ demand to set $S_1$. This implies there are $2u$ (fractional) clients served by $S_1$ which are not served by $S_2$. Let those be $C'$. Now we will make the facility $i_1$ clean by rerouting demand within facilities of $S_1$. Once $i_1$ becomes clean, we will move it to set $S_2$, as we do in Case 1. Following is the description of the cleaning process.
\begin{enumerate}
\item Let $j' \in C'$ be a client such that it is served by some $i_2 \in S_1$ and $i_2 \neq i_1$ and it is not served by any facility in $S_2$. Due to the above argument we can always find such a client.
\item We simply swap the assignments of $j$ and $j'$ between $i_1$ and $i_2$. Now, $j$ will send his demand to $i_2$ and $j'$ will send his demand to $i_1$.
\item Repeat until $i_1$ is clean.
\end{enumerate}

Repeat, the above described algorithm for set $\mathcal{S}$, until the cardinality of all the sets with non zero $z''$ values differ by at most $1$. Note that after any iteration of the above algorithm, $y''_p = y'_p = y_p$.

Now we need to upper bound the cost incurred in the cleaning process. Consider set $S \in \mathcal{S}$ and facility $i \in S$. Process of making a facility $i$ clean incurs a movement of $2u{z'}_S^p$ units of demand over at most $O(\ell)$ edges within black component $J_p$. Each such edge has length at most $d(J_p, R \setminus J_p)$. Finally, we bound the cost of rerouting $2u{z'}_S^p$ units of demand within one black component $J_p$ by $O(\ell)d(J_p, R \setminus J_p)u{z'}_S^p$.

Consider two stages of the cleaning process. The first stage ends at the moment when each set from $\mathcal{S}$ has cardinality at least $\floor{\dot{y}_p}$. Observe that each set $|S| \leq \dot{y}_p$, in the first stage, incur a cost $O(\ell) d(J_p, R \setminus J_p) {z'}_S^p u (\floor{\dot{y}_p} - |S|)$. If all sets from $\mathcal{S}$ have size at most $\ceil{\dot{y}_p}$ then we are done, otherwise stage two starts. Notice that at most $\dot{y}_p - \floor{\dot{y}_p}$ fraction of sets with cardinality $\floor{\dot{y}_p}$ will be swapped, each at most once. Using the fact that all ${z'}_S^p$ values are equal, we can bound the cost of stage two by $O(\ell)\sum_{S \in \mathcal{S} : |S| \leq \dot{y}_p} d(J_p, R \setminus J_p) {z'}_S^p u(\dot{y}_p - \floor{\dot{y}_p})$. This implies that the total cost of the procedure is bounded by $O(\ell) d(J_p, R \setminus J_p) \sum_{S \in \mathcal{S} : |S| \leq \dot{y}_p} {z'}_S^p u (\dot{y}_p - |S|) \leq O(\ell) d(J_p, R \setminus J_p) \pi(J_p) \leq O(\ell)(D(U_{J_p}) + D'(U_{J_p}))$, where the first inequality follows from Lemma~\ref{lem:massage_ineq} and the second inequality follows from Lemma~4.5 in \cite{Li16non_uniform}. This implies that the cost of the cleaning process for all black components is bounded by $O(\ell) \mathbf{LP}$.
\qed
\end{proof}

\begin{definition}
\label{complete_solution}
A solution $\{x', y', z'\}$ is called a complete solution w.r.t. a black component $p \in \Upsilon$, with set of facilities $U_p$, if the following conditions hold:

\begin{enumerate}

\item ${z'}_{S_1}^{p} = {z'}_{S_2}^{p} \, \forall S_1, S_2 \in \{S \in \mathcal{S}~|~ z_{S}^p > 0\}$.
\item ${z'}_{S}^{p} = {z'}_{S, i}^{p} \, \forall S \in \mathcal{S}$, $i \in S$.
\item ${z'}_{S, i, j}^{p} \in \{0, {z'}_{S, i}^{p}\} \, \forall S \in \mathcal{S}$, $i \in S, j \in C$.
\end{enumerate} 
\end{definition}
By splitting the sets and making copies (see Lemma 1 in \cite{spliting_svir}), we can assume that the solution $\{x', y', z'\}$ is complete w.r.t. any black component $p \in \Upsilon$. Hence the solution $\{x'', y'', z'' \}$ obtained by applying Lemma~\ref{lem:massage} will also be complete w.r.t. any black component $p \in \Upsilon$. 

\begin{lemma}
\label{lem:pre-assign}
Let $p \in \Upsilon$ be a massaged component s.t. $y_p \leq 2\ell$ and let $Z_p \in \{0 ,1\}$ be a random variable, such that $E[Z_p] = \dot{y}_p(z'') - \floor{\dot{y}_p(z'')}$. Moreover, constraints (\ref{clp:6}) to (\ref{clp:13}) are satisfied for the solution $\{x'', y'', z''\}$ and $U_p$. Then, we can pre-open a set $S \subseteq U_p$ of expected cardinality $\dot{y}_p$, where $|S| = \floor{\dot{y}_p(z'')} + Z_p$, and pre-assign a set $C' \subseteq C$ of clients to $S$ such that
	\begin{enumerate}
    	\item \label{at_most_u} each facility $i \in S$ is pre-assigned at most $u$ clients
        \item \label{out_demand} expected cost of sending not assigned demand $x_{p, C \setminus C'}$ to the root of the parent group is at most $O(\ell) d(J_p, R \setminus J_p) \pi(J_p)$
        \item \label{exp_y'_p} $\Pr[|S| = \floor{\dot{y}_p(z'')}]=\dot{y}_p(z'') - \floor{\dot{y}_p(z'')}$ and $\Pr[|S| = \ceil{\dot{y}_p(z'')}]=1 - (\dot{y}_p(z'') - \floor{\dot{y}_p(z'')})$
        \item \label{cost_of_pre} expected cost of pre-assignment and local moving of $x_{p, C \setminus C'}$ is at most $O(1)\sum_{j \in C, i \in U_p} d (i, j){x''}_{i, j}$
	\end{enumerate}
\end{lemma}

\begin{proof}
Define $\hat{z}_{S}^{p} := \frac{{z''}_{S}^{p}}{1-{z''}_{\bot}^{p}}$ and $\hat{z}_{S, i, j}^{p} := \frac{{z''}_{S, i, j}^{p}}{1-{z''}_{\bot}^{p}}$ for all sets $S \in \mathcal{S}, i \in U_p, j \in C$ and $\hat{z}_{\bot}^{p} = 0$. Observe that $\sum_{S \in \mathcal{S}} \hat{z}_{S}^{p} = 1$ by constraint (\ref{clp:6}). Recall that the sets with positive $\hat{z}_S^p$ value have cardinality either $\floor{\dot{y}_p}$ or $\ceil{\dot{y}_p}$. Define $\mathcal{S}_1 = \{S \in \mathcal{S}~|~\hat{z}_S^p > 0 \wedge |S| = \floor{\dot{y}_p}\}$ and $\mathcal{S}_2 = \{S \in \mathcal{S}~|~ \hat{z}_S^p > 0 \wedge |S| = \ceil{\dot{y}_p}\}$. If variable $Z_p$ has value 0, then we select a set from $\mathcal{S}_1$, otherwise if $Z_p = 1$ then we should select it from $\mathcal{S}_2$. To do this we proceed in the following way. If $Z_p = 0$ then we set $\bar{z}_S^p = \frac{\hat{z}_S^p}{1 - (\dot{y}_p - \floor{\dot{y}_p})}$ and $\bar{z}_{S, i, j}^p = \frac{\hat{z}_{S, i, j}^p}{1 - (\dot{y}_p - \floor{\dot{y}_p})}$ for all $S \in \mathcal{S}_1, i\in U_p, j \in C$ and $\bar{z}_S^p = 0$ and $\bar{z}_{S, i, j}^p = 0$ for $S \in \mathcal{S}_2, i \in U_p, j \in C$. Otherwise, if $Z_p = 1$ then $\bar{z}_S^p = \frac{\hat{z}_S^p}{\dot{y}_p - \floor{\dot{y}_p}}$ and $\bar{z}_{S, i, j}^p = \frac{\hat{z}_{S, i, j}^p}{\dot{y}_p - \floor{\dot{y}_p}}$ for all $S \in \mathcal{S}_2$ and $\bar{z}_S^p = 0$ and $\bar{z}_{S, i, j}^p = 0$ for $S \in \mathcal{S}_1, i \in U_p, j \in C$. Note that the expected value of $\bar{z}_S^p (\bar{z}_{S, i ,j}^p)$ is equal to $\hat{z}_S^p (\hat{z}_{S, i, j}^p)$ for each $S \in \mathcal{S}$. Moreover, $\bar{z}^p$ is a probability distribution over sets from $\mathcal{S}_1 \cup \mathcal{S}_2$, because $\dot{y}_p - \floor{\dot{y}_p} = \sum_{S \in \mathcal{S}_2} \hat{z}_S^p$ and $1 - (\dot{y}_p - \floor{\dot{y}_p}) = \sum_{S \in \mathcal{S}_1} \hat{z}_S^p$. This implies that property (\ref{exp_y'_p}) holds.

Now assume that we have selected a set $S \in \mathcal{S}_1 \cup \mathcal{S}_2$ according to this probability distribution. For each client $j \in C$ and facility $i \in S$ let $w_{i, j} = \frac{\hat{z}^p_{S, i ,j}}{\hat{z}_{S, i}} = \frac{{z''}^p_{S, i ,j}}{{z''}_{S, i}} = \frac{{z''}^p_{S, i ,j}}{{z''}_S^p}$, by constraint (\ref{clp:10}). Notice that $w_{i, j} \in [0,1]$, by constraint (\ref{clp:9}), $\sum_{i \in S} w_{i, j} \leq 1$ for every client $j \in C$ by constraint (\ref{clp:11}), and $\sum_{j \in C} w_{i,j} \leq u$ for every $i \in S$ by constraint (\ref{clp:12}). Since, the solution $\{x'', y'', z''\}$ is complete, hence  $w_{i, j} \in \{0,1\}$, for each client $j \in C$ and facility $i \in S$. Thus, $w$ defines an integral matching between $C$ and $S$ in which each client is matched at most once, and each facility from $S$ is matched at most $u$ times.

We will assign only "decided" clients, whose fractional assignment, to the black component $p$, is more than some threshold $t$ and were selected in the above described procedure. For all the other clients, we will send their demand to the root of the parent-group $G \ni p$. From Lemma~\ref{lem:group_edges} we know that the distance to the root is at most $O(\ell)d(J_p, R \setminus J_p)$.

For each $(i, j) \in S \cross C$ by $j \rightarrow i$ we denote the event that client $j$ is pre-assigned to facility $i$.

$$\mathrm{Pr}[j \rightarrow i] = \mathrm{Pr}[Z_p = 0]  \sum_{S \in \mathcal{S}_1, i \in S} \bar{z}_S^p \frac{{z''}_{S, i, j}^p}{{z''}_S^p} + \mathrm{Pr}[Z_p = 1]  \sum_{S \in \mathcal{S}_2, i \in S} \bar{z}_S^p \frac{{z''}_{S, i, j}^p}{{z''}_S^p} 
$$
$$
= \sum_{S \in \mathcal{S}_1 \cup \mathcal{S}_2, i \in S} \frac{{z''}_{S, i, j}^p}{1 - {z''}_{\bot}^p} \leq \frac{{x''}_{i,j}}{1 - {z''}_{\bot}^p}
$$ 

Now, we are ready to bound the total pre-assignment cost. Let's set a threshold $t = \frac{1}{2}$. If ${x''}_{p, j} > t$ then we say that client $j$ is decided, and not decided otherwise. Let $C''$ be the set of decided clients.

$$\mathrm{E}[\text{total pre-assignment cost}] \leq \sum_{j \in C'', i \in U_p} d(i, j) 
\mathrm{Pr}[j \rightarrow i]  \leq  \sum_{j \in C'', i \in U_p} \frac{d(i, j) {x''}_{i, j}}{1 - {z''}_{\bot}^p}$$


All the rest of the clients which were not assigned we will send to the root of a group via their respective representatives (each representative $v$ is co-located with facilities from $U_v$). The cost of sending ${x''}_{i,j}$ units of not pre-assigned demand to their representatives we can bound by ${x''}_{i,j}d(i,j)$. Notice that client can be either pre-assigned or sent to the root, so in the worst case we can bound the total cost of pre-assignment and moving to representatives by

$$\sum_{j \in C, i \in U_p} \frac{d(i, j){x''}_{i, j}}{1 - {z''}_{\bot}^p}$$

By $j \rightarrow$ we denote the event that client $j$ was assigned and $\Pr[j \rightarrow]$ the probability of this event. We will apply the following bound only for decided clients because ${x''}_{p, j} \geq \frac{2}{\ell} \geq {z''}_{\bot}^p$.
$$\mathrm{P}[j\rightarrow] \geq \frac{{x''}_{p, j} - {z''}_{\bot}^p}{1 - {z''}_{\bot}^p}$$ It might happen that a decided client $j$ is not assigned to his black component $p$. Then we send demand ${x''}_{p, j}$ to the root of his group. We upper bound the cost of this step in the following way.

\begin{eqnarray}
\label{ineq_cost_out_1}
\sum_{j \in C''} (1 - \mathrm{P}[j \rightarrow]) {x''}_{p, j} O(\ell) d(J_p, R \setminus J_p) 
\leq O(\ell) d(J_p, R \setminus J_p) \sum_{j \in C''} {x''}_{p, j}(1 - {x''}_{p, j}).&
\end{eqnarray}
Notice that decided client can be assigned in at most one black component. As we mentioned before, we will send all demand of not decided clients to the root. Hence, the cost for sending this demand is at most

\begin{eqnarray}
\label{ineq_cost_out_2}
O(\ell) d(J_p, R \setminus J_p) \sum_{j \in C \setminus C''} {x''}_{p, j} & \leq & O(\ell) d(J_p, R \setminus J_p) \sum_{j \in C \setminus C''} {x''}_{p, j} (1 - {x''}_{p, j})
\end{eqnarray}
where the inequality follows from the fact that client $j$ is not decided, so $\frac{1}{2} \leq 1 - x_{p, j}$. By summing up inequalities (\ref{ineq_cost_out_1}) and (\ref{ineq_cost_out_2}) and by the fact that ${x''}_{p, j} = x_{p, j}$ we can bound the cost of sending not assigned demand from the representatives to the root by $O(\ell) d(J_p, R \setminus J_p) \pi(J_p)$.
\qed
\end{proof}

\begin{lemma}
\label{lem:blue}
For any group $\mathcal{G}$, the total demand is at most $\sum_{p \in \mathcal{G} \setminus \{v_{\mathcal{G}}\}} u \dot{y}_p + (1+O(1/\ell)) u \dot{y}_{v_{\mathcal{G}}}$. 
\end{lemma}

\begin{proof}
First let us compare the opening value $\dot{y}_p z_{\bot}^{p}$ to the original opening $\sum_{i \in U_p} z_{\bot, i}^{p}$.
\begin{eqnarray}
\label{ineq:blue_opening}
\frac{\dot{y}_p z_{\bot}^{p}}{\sum_{i \in U_p} z_{\bot, i}^{p}} \leq \frac{2 \ell z_{\bot}^{p}}{\ell^2 z_{\bot}^{p}} = \frac{2}{l}
\end{eqnarray} 
The inequality follows from the fact that $\dot{y}_p \leq 2\ell$ and from Constraint~(\ref{clp:13}).
By adding Constraint~(\ref{clp:12}) for set $\bot$ over all $i \in U_p$, we get $\sum_{i \in U_p}z_{\bot, i}^{p}u \geq \sum_{i \in U_p, j \in C}z_{\bot, i, j}^{p}$. Consider the demand served by the $\bot$ set in the initial LP solution and the blue opening sent from $p$ to $v_{\mathcal{G}}$ which is equal to $\sum_{i\in U_p} z_{\bot, i}^{p} - \dot{y}_p z_{\bot}^{p}$. Now we show that the ratio between this demand and the capacity corresponding to the blue opening is at most $(1+O(1/\ell))$. Note that the demand served by $\bot$ remains the same in all the modified solutions.
$$\frac{\sum_{i \in U_p, j \in C}z_{\bot, i, j}^{p}}{u \left( \sum_{i \in U_p} z_{\bot, i}^{p} - \dot{y}_{p} z_{\bot}^p \right)} \stackrel{Const. (\ref{clp:12})} \leq \frac{u \sum_{i \in U_p} z_{\bot, i}}{u \left(\sum_{i \in U_p} z_{\bot, i}^{p} - \dot{y}_{p}z_{\bot}^p\right)} = \frac{1}{1-\frac{ \dot{y}_{p}z_{\bot}^p}{\sum_{i \in U_J} z_{\bot, i}^{p}}} 
\stackrel{Ineq. (\ref{ineq:blue_opening})} \leq \frac{1}{1-2/\ell} = 1+O(1/\ell)$$

The following calculations show that the capacity violation on $v_{\mathcal{G}}$ w.r.t. to the total demand served by $\bot$ set of each $U_p : p \in \mathcal{G} \setminus \{v_{\mathcal{G}}\}$ and the blue opening sent from all the black components of that group is at most $(1+O(1/\ell))$.
\begin{eqnarray}
\label{ineq:blue_opening_2}
\frac{\sum_{p \in \mathcal{G} \setminus \{v_{\mathcal{G}}\}} \sum_{i \in U_p, j \in C}z_{\bot, i, j}^{p}}{\sum_{p \in \mathcal{G} \setminus \{v_{\mathcal{G}}\}} u \left( \sum_{i \in U_p} z_{\bot, i}^{p} - \dot{y}_{p}z_{\bot}^p \right)} \leq
 \max_{p \in \mathcal{G} \setminus \{v_{\mathcal{G}}\}} \frac{\sum_{i \in U_p, j \in C}z_{\bot, i, j}^{p}}{\sum_{i \in U_p} z_{\bot, i}^{p} - \dot{y}_{p}z_{\bot}^p} \leq 1+O(1/\ell)
\end{eqnarray}
 
Now we compare the rest of the demand $\sum_{S \in \mathcal{S}, i \in S, j \in C} \hat{z}_{S, i, j}^p$ 
to the rest of the opening $\dot{y}_p$ for each black component $p \in \mathcal{G} \setminus \{v_{\mathcal{G}}\}$ and show that this demand is at most $u \dot{y}_p$. 
$$
\sum_{S \in \mathcal{S}, i \in S, j \in C} \hat{z}_{S, i, j}^p = \frac{1}{1-{z''}_{\bot}^p}\sum_{S \in \mathcal{S}, i \in S, j \in C} {z''}_{S, i, j}^p \stackrel{Cons.  (\ref{clp:12})} \leq \frac{1}{1-{z''}_{\bot}^p}\sum_{S \in \mathcal{S}, i \in S} u {z''}_{S, i}^p = $$ 
$$\frac{1}{1-{z'}_{\bot}^p}\sum_{S \in \mathcal{S}, i \in S} u {z'}_{S, i}^p = \frac{1}{1-z_{\bot}^p}\sum_{S \in \mathcal{S}, i \in S} u {z}_{S, i}^p \stackrel{Cons. (\ref{clp:10})}= \frac{1}{1-z_{\bot}^p}\sum_{S \in \mathcal{S}} u |S| {z}_{S}^p = u \dot{y}_p
$$

Adding the above inequalities for each $p \in \mathcal{G} \setminus \{v_\mathcal{G}\}$ and the inequality (\ref{ineq:blue_opening_2}) proves the lemma.
\qed
\end{proof}

\begin{lemma}
	\label{lem:massage_ineq}
	Consider a black component $p \in \Upsilon$ and its set of representatives $J_p$. We have $$\sum_{S \in \mathcal{S} : |S|\leq \dot{y}_p}{z'}_S^p u (\dot{y}_p - |S|) \leq 2\pi(J_p).$$ 
\end{lemma}

\begin{proof}
$$\sum_{S \in \mathcal{S} : |S|\leq \dot{y}_p}{z'}_S^p u (\dot{y}_p - |S|) \leq \sum_{S \in \mathcal{S}}{z'}_S^p u |y_p - |S|| = $$
    
$$\sum_{S \in \mathcal{S}} | {z'}_S^p \sum_{j \in C}x_{p,j} - \sum_{j \in C}{x'}_{S,j}| \leq
\sum_{j \in C} \sum_{S \in \mathcal{S}} |{z'}_S^p(x_{p,j} - \frac{{x'}_{S,j}}{{z'}_S^p})| =$$

$$\sum_{j \in C} \big[ \sum_{S \in \mathcal{S} : {x'}_{S, j} = 0} {z'}_S^p x_{p, j} + \sum_{S \in \mathcal{S} : {x'}_{S, j} = {z'}_S^p} {z'}_S^p(1 - x_{p, j}) \big] \leq 2 \sum_{j \in C} x_{p,j} (1 - x_{p,j}) = 2 \pi(J_p)$$
\qed
\end{proof}

If $|C| = ku$ then we can improve the connection cost of our algorithm. We omit the proof of the following theorem, because it is similar to the proof of the Theorem \ref{thm:general}.

\begin{theorem}
\label{thm:case}
There is a bi-factor randomized rounding algorithm for hard uniform capacitated $k$-median problem, for $|C| = ku$, with $O(1/\epsilon)$-approximation under $1+\epsilon$ capacity violation.
\end{theorem}
\end{document}